\documentclass[aps,pra,reprint,twocolumn,superscriptaddress,nofootinbib]{revtex4-1}
\pdfoutput=1

\usepackage[utf8]{inputenc}
\usepackage[english]{babel}
\usepackage[T1]{fontenc}
\usepackage{amsmath}
\usepackage{bm}
\usepackage{amsfonts}
\usepackage{comment}
\usepackage{hyperref}
\allowdisplaybreaks

\usepackage[most]{tcolorbox}

\usepackage{tikz}
\usepackage{braket}
\usepackage{natbib}
\usepackage{amsthm}

\newtheorem{thm}{Theorem}
 \newtheorem{proposition}{Proposition}[thm]
 
 \newtheorem{lemma}{Lemma}[thm]
 \newtheorem{fact}{Fact}[thm]
\newtheorem{cor}{Corollary}[thm]

\newtheorem{alg}{Algorithm}[thm]

\DeclareMathOperator{\Tr}{tr}

\newcommand{\norm}[1]{\left\lVert#1\right\rVert}
\newcommand{\vertiii}[1]{{\left\vert\kern-0.25ex\left\vert\kern-0.25ex\left\vert #1 
    \right\vert\kern-0.25ex\right\vert\kern-0.25ex\right\vert}}

\newcommand{\vct}[1]{\mathbf{#1}}
\renewcommand{\vec}{\bm}

\newcommand{\BE}{\mathbb{E}}
\newcommand{\CF}{\mathcal{F}}

\newcommand{\CO}{\mathcal{O}}

\newcommand{\CU}{\mathcal{U}}
\newcommand{\CV}{\mathcal{V}}

\newcommand{\lV}{\lVert}
\newcommand{\rV}{\rVert}
\newcommand{\e}{\mathrm{e}}
\newcommand{\iunit}{\mathrm{i}}

\newcommand{\AC}[1]{
{\color{black} #1}
}

\newcommand{\Last}[1]{
{\color{black} #1}
}

\usepackage[noend]{algpseudocode}
\usepackage{algorithm,algorithmicx}

\algrenewcommand\alglinenumber[1]{\sf\scriptsize\color{blue}{#1}}
\algrenewcommand\algorithmicrequire{\textbf{Input:}}
\algrenewcommand\algorithmicensure{\textbf{Output:}}

\begin{document}
\title{Concentration for random product formulas}

\date{\today}
\author{Chi-Fang Chen}
\thanks{These authors contributed equally} \email{email:chifang@caltech.edu}
\affiliation{Department of Physics, Caltech, Pasadena, CA, USA}
\author{Hsin-Yuan Huang}
\thanks{These authors contributed equally} \email{email:chifang@caltech.edu}
\affiliation{Institute for Quantum Information and Matter, Caltech, Pasadena, CA, USA}
\affiliation{Department of Computing and Mathematical Sciences, Caltech, Pasadena, CA, USA}
\author{Richard Kueng}
\affiliation{Institute for Quantum Information and Matter, Caltech, Pasadena, CA, USA}
\affiliation{Department of Computing and Mathematical Sciences, Caltech, Pasadena, CA, USA}
\affiliation{Institute for Integrated Circuits, Johannes Kepler University Linz, Austria}
\author{Joel A.~Tropp}
\affiliation{Department of Computing and Mathematical Sciences, Caltech, Pasadena, CA, USA}

\begin{abstract}
 Quantum simulation has wide applications in quantum chemistry and physics.
    Recently, scientists have begun exploring the use of randomized methods for accelerating quantum simulation.
    Among them, a simple and powerful technique, called \textsc{qDRIFT}, is known to generate random product formulas
    for which the \textit{average} quantum channel approximates the ideal evolution.
    \AC{\textsc{qDRIFT} achieves a gate count that does not explicitly depend on the number of terms in the Hamiltonian, which contrasts with Suzuki formulas.
    This work aims to understand the origin of this speed-up by comprehensively analyzing a \textit{single realization} of the random product formula produced by \textsc{qDRIFT}.
    The main results prove that a typical realization of the randomized product formula approximates the ideal unitary evolution up to a small diamond-norm error.
    The gate complexity is already independent of the number of terms in the Hamiltonian, but it depends on the system size and the sum of the interaction strengths in the Hamiltonian.
    Remarkably, the same random evolution starting from an arbitrary, but fixed, input state yields a much shorter circuit suitable for that input state. In contrast, in deterministic settings, such an improvement usually requires initial state knowledge. 
    The proofs depend on concentration inequalities for vector and matrix martingales, and the framework is applicable to other randomized product formulas. Our bounds are saturated by certain commuting Hamiltonians.}

\end{abstract}
    
\maketitle

\section{Introduction}

Simulating complex quantum systems is one of the most promising applications for quantum computers.  This task has many applications, such as developing new pharmaceuticals, catalysts, and materials \cite{georgescu2014quantum, babbush2018low, mcardle2020quantum}, as well as solving linear algebra problems \cite{subacsi2019quantum, huang2019near, an2019quantum}. \AC{The task of digital quantum (dynamics) simulation can be phrased as a compiling problem: 
approximate a given unitary, say a Hamiltonian evolution $U = \mathrm{e}^{-\iunit  Ht}$, by a product of `simple' unitaries $g_k$:
\begin{equation}
U=\mathrm{e}^{-\iunit  H t} \approx V = g_1 \cdots g_N. \label{eq:compilation}
\end{equation}
Such a decomposition into elementary gates should obey two conditions: (i) It should accurately approximate the target unitary.
In this work, we require that the error in operator norm\footnote{\AC{For measuring time-evolved observables $\Tr(\rho(t)O)$, the operator norm suffices; for quantum phase estimation of ground state energies things can get more complicated, though. There, the estimates may depend on the details and vary orders of magnitude~\cite{google2020THC} and the 
assessment of real costs is an on-going research direction that is beyond the scope of this work.}}
satisfies $\|U-V \| \leq \epsilon$ for some specified accuracy parameter $\epsilon$.
Moreover, (ii) the decomposition should cost as little as possible.  Several cost functions make sense in this context, but we will focus on the gate complexity, i.e., the number $N$ of simple gates\footnote{\AC{In this work, we will count $\e^{-\iunit h_jt}$ as a single gate, as in~\cite{childs2019faster,childs2019theory,berry2020time}. Strictly speaking, in a fault-tolerant quantum computer, we would further decompose each $\e^{-\iunit h_jt}$ into universal 2-qubit gates.
According to the Solovay--Kiteav Theorem~\cite{dawson2005kitaev} this incurs at most a constant (multiplicative) overhead, but the exact cost depends on the type of hardware. Note we do not discuss the query complexity that appears in LCU approaches~\cite{martyn2021grand}. 
}
} on the right-hand side of Eq.~\eqref{eq:compilation}.
}

The earliest compilation procedures for quantum simulation were based on \textit{product formulas} \cite{suzuki1991general,lloyd1996universal}, also known as \textit{Trotterization}, or \textit{splitting methods}. They depend on the idea of approximating the matrix exponential of a sum by a product of matrix exponentials.

\AC{We will review one such construction below in Eq.~\eqref{eq:first_order}. 
Subsequently, alternative principles have led to the development of other quantum simulation algorithms. 
These include linear combination of unitaries~\cite{LCU}, quantum signal processing~\cite{low2017optimal} and qubitization~\cite{Low_2019_qubitize}. 
By and large, these algorithms rely on more powerful quantum computing primitives to yield improved performance in accuracy and cost. We refer to~\cite{martyn2021grand} for a systematic review. 

Despite these advanced simulation techniques, product formulas
have recently undergone a renaissance~\cite{berry2015hamiltonian,low2017optimal,childs2019faster,childs2019theory,berry2020time,campbell2019random}. 
They are not only simple and (comparatively) easy to implement on near-term devices, but they also remain very competitive~\cite{childs2019theory} provided that they incorporate information about the structure of the problem, such as initial state knowledge~\cite{low_energy2020}, locality~\cite{Haah_2021} or the commutator structure of Hamiltonian~\cite{childs2019theory}. The purpose of this paper is to explore the randomized aspects of constructing product formulas. }

\AC{
We begin by reviewing the first-order Lie--Trotter formula, which will be later contrasted with a randomized variant (\textsc{qDRIFT}). To that end, consider a quantum many-body Hamiltonian  $H = \sum_{j=1}^L h_j$ composed of $L$ simple terms $h_j$. The first-order formula cycles through each term in the Hamiltonian
\begin{equation}
    U \approx \big(\exp \left( -\mathrm{i} (tL / N) h_L \right) \cdots \exp \left( -\mathrm{i} (tL / N) h_1 \right)\big)^{N/L}.\label{eq:first_order}
\end{equation}
}
\AC{
 To approximate the target unitary $U$ up to accuracy $\epsilon$ in operator norm, a total gate count $N = \mathcal{O}(L \lambda^2 t^2 / \epsilon)$ suffices\footnote{\AC{To finally compile into universal gates, run a standard gate synthesis algorithm (such as Solovay--Kiteav) for each simple term. This yields another multiplicative factor on the gate count.}} \cite[Section~3]{childs2019theory}.  In this expression, $t$ is the simulation time, $L$ is the number of terms, and $\lambda = \sum_j \norm{h_j} $ summarizes the interaction strengths within $H$.
 The main idea is to cancel out the leading-order term in the Taylor expansion by including each of the $L$ terms. Owing to this construction, the factor $L$ remains in higher-order Suzuki formulas where the gate count $N = \mathcal{O}(L (\lambda t)^{1+ o(1)} / \epsilon^{o(1)})$, even though the time-dependence becomes nearly optimal. We refer to  Table~\ref{table:qDRFIT_vs_suzuki} for a sharper gate count incorporating commutators.
Recently, researchers started using randomization to improve the performance of product formulas \cite{childs2019faster, campbell2019random, ouyang2020compilation,faehrmann2021randomized}.  Campbell~\cite{campbell2019random} introduced the \textsc{qDRIFT} algorithm,

which approximates the target evolution $U=\e^{-\iunit Ht}$ by a quantum channel that results from averaging products $V_N \cdots V_1$ of random unitaries. Each $V_k$ corresponds to a short-time evolution based on a single term $h_K$ from the Hamiltonian.
The index $K$ is selected randomly, according to an importance sampling distribution $(p_1, \dots, p_L)$, constructed to match the leading order of time step
\begin{equation*}
\mathbb{E} \left[ V_k \right]
\approx \exp \big( - \mathrm{i} (t/N)\mathbb{E}\left[h_K/p_K \right] \big) = U^{1/N}.
\end{equation*}
This approximation is achieved by averaging over a single (random) unitary. 
(In contrast, the first order Suzuki formula \eqref{eq:first_order} must cycle through all terms which incurs an extra $L$-factor.) Independence among the $V_k$'s then ensures

\begin{equation*}
\mathbb{E} \left[ V_N \cdots V_1 \right] = \mathbb{E} \left[ V_N \right] \cdots \mathbb{E} \left[ V_1 \right] \approx \big(U^{1/N}\big)^N = U.
\end{equation*}
Campbell proved that the averaged Hamiltonian approximates the true Hamiltonian
when the gate count satisfies
\begin{align}
    N = \mathcal{O}(\lambda^2 t^2 / \epsilon). \label{eq:qDRIFT_gate_count}
\end{align}
Observe that the gate count is \textit{independent} of the number of terms $L$ in the Hamiltonian.
A summary of this procedure is as follows.
\begin{alg}[\textsc{qDRIFT}]\end{alg}\label{alg:qdrift}
\begin{tcolorbox}[colback=black!5!white,colframe=black!75!black,arc=0mm]
\begin{flushleft}
\begin{itemize}
\item[]\textbf{Inputs:} Hamiltonian $H=\sum_{j=1}^L h_j$ with interaction strength $\lambda = \sum_j \|h_j \| $, evolution time $t$, and number of steps $N$.

\item[]\textbf{At each $t/N$ interval:}
evolve a random term in Hamiltonian 
\begin{align}
    &V_k = \exp (-\mathrm{i} (t/N) X_k)\label{eq:qDRIFT} 
\end{align}

according to its importance 
\begin{align*}
& X_k \overset{\textit{i.i.d.}}{\sim} X= 
\begin{cases}
\tfrac{\lambda}{\|h_1 \|}h_1 & \text{with prob. } p_1 = \tfrac{\|h_1 \|}{\lambda} \\
& \vdots \\
 \tfrac{\lambda}{\|h_L\|}h_L & \text{with prob. } p_L = \tfrac{\|h_L\|}{\lambda}
\end{cases}. 
\end{align*}

\item[]\textbf{Output:} the unstructured (randomly generated) product formula
\begin{equation*}
V^{(N)} = V_N \cdots V_1.
\end{equation*}
\end{itemize}

\end{flushleft}
\end{tcolorbox}

}







\AC{

Operationally, Campbell considers a black box that applies a new random product $V_N \cdots V_1$ of unitaries every time it is invoked. The average of all possible product formulas is $\mathcal{V}^{(N)}(X) = \BE[V_N \cdots V_1 X V_1^\dagger \cdots V_N^\dagger]$ and forms a completely positve trace-preserving (CPTP) map. The ideal unitary also forms a CPTP map given by $\mathcal{U}(X) = U X U^\dagger$.
Campbell proves that \eqref{eq:qDRIFT_gate_count} ensures that the two CPTP maps are $\epsilon$-close in diamond distance. 

It is interesting to compare the gate count~\eqref{eq:qDRIFT_gate_count} achieved by \textsc{qDRIFT} with very recent and powerful results for (deterministic) product formulas~\cite{childs2019theory,Haah_2021}, see Table~\ref{table:qDRFIT_vs_suzuki}. Broadly speaking, deterministic product formulas can achieve gate counts that are linear in time $t$ and the number $L$ of terms. The \textsc{qDRIFT} gate count, on the other hand, is independent of $L$, but quadratic in $t$. This implies that \textsc{qDRIFT} should be favored for short-time simulations rather than long-time simulations. For systems with many terms in the Hamiltonian ($L \gg 1$), such as in quantum chemistry and the SYK model \cite{sachdev1993gapless, polchinski2016spectrum, maldacena2016remarks}\footnote{ While \textsc{qDRIFT} provides performance guarantee on rather flexible choices of models, there is a specialized quantum algorithm for simulating the SYK models using much fewer gates~\cite{Babbush_2019}.}, there are indeed physically relevant time-scales for which \textsc{qDRIFT} outperforms Suzuki formulas; see Table~\ref{table:qDRFIT_vs_suzuki}.
}

\begin{table}[t]\label{table:qDRFIT_vs_suzuki}
{
\begin{tabular}{lll}
\hline
qDRIFT      \ \ \ \ \ \                & $p$-th order Suzuki~\cite{childs2019theory}   \ \ \                    & Systems               \\\hline
$\lambda^2 t^2/\epsilon $&$L\vertiii{H}_1 \lambda^{\frac{1}{p}} t^{1+\frac{1}{p}}/\epsilon^{\frac{1}{p}}$ & general \\\hline
$n^2t^2/\epsilon$ & $n^{1+\frac{1}{p}}t^{1+\frac{1}{p}}\epsilon^{\frac{1}{p}}$ & nearest neighbor \\\hline
&&power-law($1/r^\alpha$) \\
$n^2t^2/\epsilon$  & $(nt)^{1+\frac{1}{p}+\frac{d}{\alpha-d}}/\epsilon^{\frac{d}{\alpha-d}+\frac{1}{p}}$      & $2d \le \alpha$  \\
$n^2t^2/\epsilon$    & $n^{2+\frac{1}{p}}t^{1+\frac{1}{p}}/\epsilon^{\frac{1}{p}}$                       & $d \le \alpha \le 2d$ \\
$n^{4-2\frac{\alpha}{d}}t^2/\epsilon$        & $n^{3-\frac{\alpha}{d}+(2-\frac{\alpha}{d})/p}t^{1+\frac{1}{p}}/\epsilon^{\frac{1}{p}}$            & $0 \le \alpha \le d$  \\\hline
&&$k$-local\\
$n^{k+1}t^2/\epsilon$ & $n^{\frac{3k-1}{2}+\frac{k+1}{2p}}t^{1+\frac{1}{p}}/\epsilon^{\frac{1}{p}}$ &  $\lV h_j\rV=
\CO \left(\sqrt{1/n^{k-1}}\right)$ \\
$n^{2}t^2/\epsilon$ & $n^{k+\frac{1}{p}}t^{1+\frac{1}{p}}/\epsilon^{\frac{1}{p}}$ &   $\lV h_j\rV=
\CO \left(1/n^{k-1}\right)$
\end{tabular}}
\caption{
\AC{
\emph{Error bound comparison \textsc{qDRIFT} vs.\ $p$-th order Suzuki formulas~\cite{childs2019theory}.} {We consider systems with spatially local interaction in $d$-dimension, power-law interaction in $d$-dimension, and all-to-all $k$-local interaction. 
Mutiplicative overhead depending only on $p,k,\alpha,d,\log(n)$ are dropped, and practically we may think of $p=4$ or $6$.} The $t^2$-dependence sets a time scale before which \textsc{qDRIFT} yields an advantage. For geometrically local Hamiltonians, the Suzuki formulas are effectively tight~\cite{Haah_2021} (Exploiting locality, \cite{Haah_2021} further removes the $o(1)$ dependence), while \textsc{qDRIFT} already performs poorly after very short times ($t\sim \CO(1/n)$). However, once $L$ gets very large, such as in k-local models, there are physically relevant time scales $t=\CO( n^{(k-3)/2})$ where \textsc{qDRIFT} becomes advantageous.
Note that normalization conventions do affect the gate count substantially. For systems with random coefficients (e.g. SYK-like), one typically fixes $\sum_j \norm{h_j}^2=n$~\cite{sachdev1993gapless, chen2019operator}. But other conventions, like fixing $\lambda = \sum_j \norm{h_j}=n$, are also widespread~\cite{Lashkari_2013}. 
We refer to Appendix~\ref{sec:calc_table} for detailed calculations.}
}
\end{table}

\begin{figure*}
    \centering
    \includegraphics[width=0.8\textwidth]{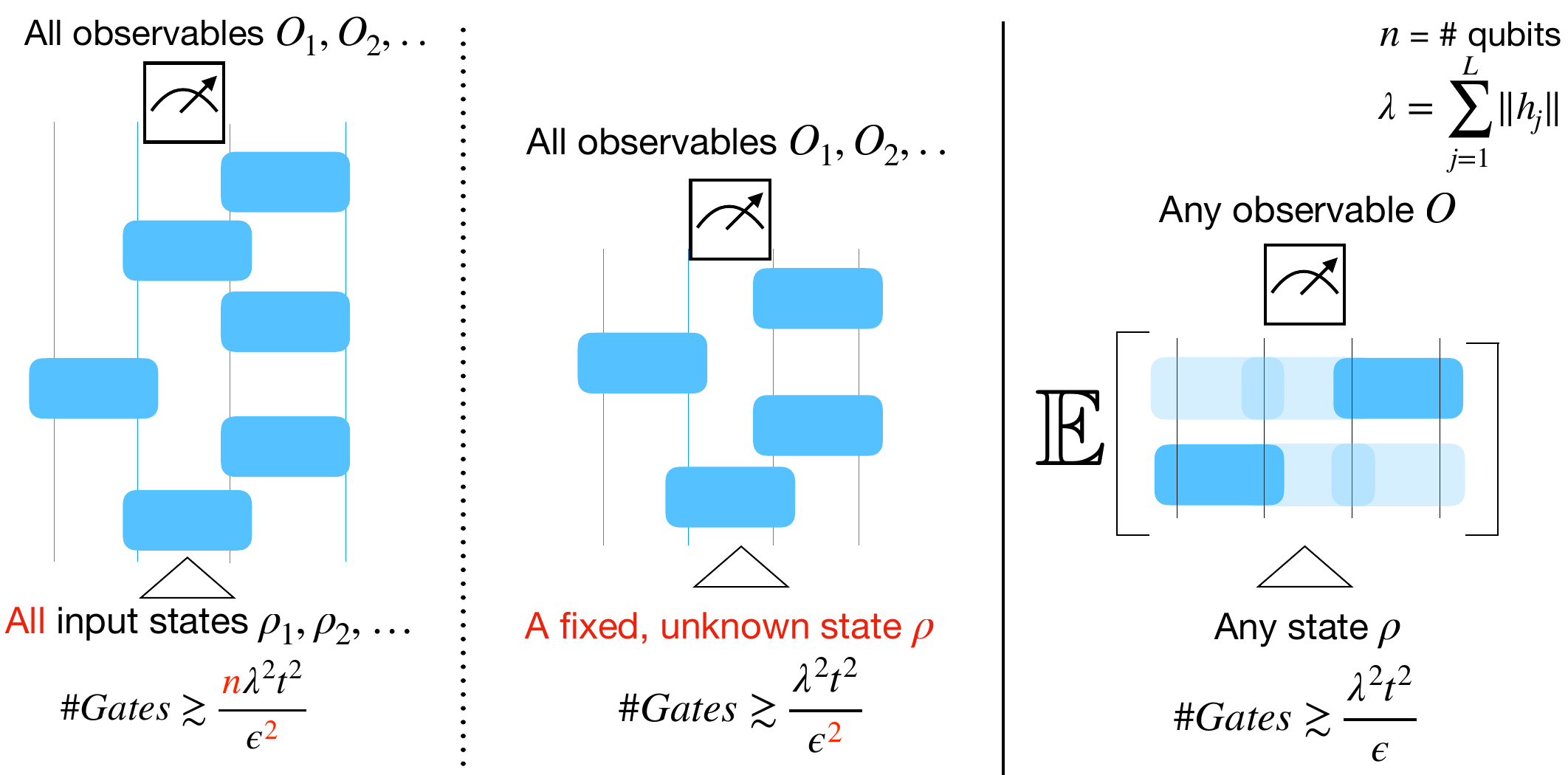}
    \caption{\AC{
     A pictorial summary of the main results. (left) To sample a product formula that works for all $n$-qubit input states and observables \textit{with high probability}, the number of gates is larger than sampling a product formula that works for a fixed, yet arbitrary, input state (center).  Resampling fresh product formulas every time (right) produces an average channel that requires even fewer gates; this is the original \textsc{qDRIFT} guarantee~\cite{campbell2019random}. }
    }
    \label{fig:mainresult}
\end{figure*}
\AC{
To summarize, there are interesting quantum simulation problems where the \textsc{qDRIFT} gate count~\eqref{eq:qDRIFT_gate_count} compares favorably with very recent and powerful results about high-order Suzuki formulas. This advantage is a consequence of \textit{randomization}.
But we may ask whether the benefit arises from the random choice of an individual product formula or whether it is due to the averaging of many product formulas together.  Which aspect produces the speed-up?

In this work, we analyze why random product formulas are efficient.To do so, we study the performance of \textit{a random instance} of the product formula $V_N \cdots V_1$.  By contrasting this instance with the average channel $\mathcal{V}^{(N)}(X)$, we deduce that random sampling alone allows us to avoid the dependency on the number $L$ of terms in the Hamiltonian; it is not necessary to average many different product formulas.  Our results establish strong concentration bounds for two cases, depicted in Figure~\ref{fig:mainresult}.}

First, let $n$ denote the number of system constituents, e.g., qubits.
If the gate count $N$ obeys
\begin{equation*}
N \geq \Omega( n t^2 \lambda^2 / \epsilon^2),  
\end{equation*}
then, {with high probability}, a single realization of the random product formula
approximates the ideal target unitary up to accuracy $\epsilon$ in operator norm.

By the probabilistic method, this result also establishes, for the first first time, the \textit{existence} of product formulas whose gate count $N$ is independent of the number $L$ of terms in the Hamiltonian but depends on the system size $n$ instead.  On the other hand, we cannot easily verify the quality of any given instance.

In practice,  we often wish to evolve a fixed input quantum state $\rho$, which may be arbitrary and unknown.This change in the problem statement has profound implications for randomized quantum simulation.
With high probability, a random product formula with
\begin{equation}
N \geq \Omega( t^2 \lambda^2 / \epsilon^2)    \label{eq:fixed-input-gate-count}
\end{equation}
terms suffices to achieve an $\epsilon$-approximation $V_N \cdots V_1 \rho V_1^\dagger \cdots V_N^\dagger$ of the ideal time-evolved state $U \rho U^\dagger$ with respect to trace distance. Roughly speaking, this result implies that each input state has a set of product formulas that are $n$ times shorter than a ``general-purpose'' product formula that works for all input states simultaneously.
Although the set of effective product formulas depends on the choice of state and observable,
the formulas are all produced by the same randomized procedure.  Remarkably, this procedure
does not exploit any knowledge of the particular input state. 

\AC{Note that the gate count required for the original \textsc{qDRIFT} protocol~\eqref{eq:qDRIFT_gate_count} and Rel.~\eqref{eq:fixed-input-gate-count} only differ in the scaling with accuracy: order $1/\epsilon$ for the full \textsc{qDRIFT} channel vs.\ order $1/\epsilon^2$ for a single random product formula with fixed input.\footnote{ An anonymous referee pointed out a potential connection with Stochastic Differential Equation solvers in the \Last{Euler-Maruyama Scheme}. There, given a sample $X_{sample}$ produced by the solver, there are also two ways to measure the error: strong error $\BE|X - X_{sample} |$ and  weak error $|\BE X - \BE X_{sample} |$. Interestingly, the convergence order is\Last{ $1/2$ and $1$}, respectively, seemingly resembling the discussion here on \textsc{qDRIFT}.} This discrepancy is reminiscent of the \textit{mixing Lemma} by Hastings~\cite{hastings2016turning} and Campbell~\cite{campbell_mixing16}: mixing unitaries yields a ``quadratic" improvement in accuracy.   See Lemma~\ref{lem:mixing_maintext} below for a statement.

In this work, we analyze several classes of randomized product formulas, including \textsc{qDRIFT} as a particular example.   The underlying theory is far more general because it relies on powerful concentration results for matrix and vector martingales.  The approach yields strong concentration results for any product of random unitary matrices, such as the ones that arise from randomized compiling~\cite{hastings2016turning,campbell_mixing16}.   We are confident that these ideas are applicable to other problems that involve functions of random matrices, such as~~\cite{chen2021optimal,chen2021concentration}.}

\paragraph*{\textbf{Roadmap}}
The rest of this paper is organized as follows.
Section~\ref{sec:results} presents the main theoretical contributions to this work in detail. These are further supported and illustrated by numerical experiments presented in Section~\ref{sec:numerics}.
Section~\ref{sec:proofidea} contains two instructive examples, as well as a non-technical illustration of the underlying proof idea. We introduce related background for martingales in Appendix~\ref{sec:intro_martingale}. Details and rigorous arguments are provided in Appendix~\ref{sec:proofs}).
The final appendix section  (Section~\ref{sec:tightness}) establishes
asymptotic tightness for time-evolving two simple (commuting) Hamiltonians.

\section{Main results} \label{sec:results}

This section gives rigorous results for the error incurred by a randomly sampled product formula $V_N \cdots V_1$,
as compared with the ideal unitary evolution operator $U = \exp(-\mathrm{i} H t)$.
The results depend on the distance measure and the particular setup,
which we discuss separately in the following subsections.

\subsection{Error bound in diamond distance: Worst-case error over all input states and observables}

The diamond distance is a standard distance measure for quantum channels. To compare two unitaries $U_1$ and $U_2$,
the diamond distance is equivalent to
\begin{align*}
    \mathrm{dist}_{\diamond}(U_1, U_2 ) &= \max_{\ket{\psi}: \text{ state}} \norm{U_1 \ket{\psi}\!\bra{\psi} U_1^\dagger - U_2 \ket{\psi}\!\bra{\psi} U_2^\dagger}_{1} \\
    &= \max_{\ket{\psi}: \text{ state}} \,\,\, \max_{O: \norm{O} \leq 1} \left| \langle O \rangle_{U_1 \ket{\psi}} -  \langle O \rangle_{U_2 \ket{\psi}} \right|,
\end{align*}
where $\norm{\cdot}_1$ is the trace norm and $\langle O \rangle_{\ket{\phi}} = \bra{\phi} O \ket{\phi}$ is the expectation value of an observable $O$ for the quantum state $\ket{\phi}$.

Operationally, this means that the expectation value of $O$ evaluated on the output state would differ at most by the diamond distance between $U_1, U_2$ for any input quantum state $\ket{\psi}$ and any observable $O$ with eigenvalues in $[-1, 1]$.

Our first main result bounds the gate complexity that suffices to guarantee that the randomly sampled product formula $V_N \cdots V_1$ is close to the ideal evolution $\exp(-\mathrm{i} t H)$ in this \textit{worst-case} error metric.

\begin{thm}[qDRIFT: Gate complexity for small diamond distance]\label{thm:allinput}
Consider an $n$-qubit Hamiltonian $H = \sum_j h_j$ with $\lambda = \sum_j \norm{h_j}$.
Draw a randomized product formula $V_N \cdots V_1$ from \eqref{eq:qDRIFT} with gate count
 \begin{equation}
 \label{eq:randomgate}
 N \geq \Omega\left((n+ \log(1 / \delta)) t^2 \lambda^2 / \epsilon^2 \right).
 \end{equation}
With probability at least $1 - \delta$, the diamond distance error satisfies
 \begin{equation*}
     \max_{\ket{\psi}: \text{ state}} \,\,\, \max_{O: \norm{O}  \leq 1} \left| \langle O \rangle_{V_N \cdots V_1 \ket{\psi}} -  \langle O \rangle_{\exp(-\mathrm{i} t H) \ket{\psi}} \right| < \epsilon.
 \end{equation*}
\end{thm}

To keep notation as simple as possible, we have formulated this result for pure input states $|\psi \rangle$.Convexity readily allows for extending the bound to mixed input states $\rho = \sum_i p_i |\psi_i \rangle \! \langle \psi_i|$ as well. A complementary result bounds the expected approximation error in terms of gate count $N$.

\begin{cor}[qDRIFT: Error bound in diamond distance]
\label{cor:exp_allinput}
Consider an $n$-qubit Hamiltonian $H = \sum_j h_j$ with $\lambda = \sum_j \norm{h_j}$.
A randomized product formula $V_N \cdots V_1$, drawn from \eqref{eq:qDRIFT}, 
has expected diamond-distance error
\begin{align*}
& \mathbb{E} \left[ \max_{\ket{\psi}: \text{ state}} \,\,\, \max_{O: \norm{O}  \leq 1} \left| \langle O \rangle_{V_N \cdots V_1 \ket{\psi}} -  \langle O \rangle_{\exp(-\mathrm{i} t H) \ket{\psi}} \right| \right] \\
\lesssim &\sqrt{\frac{n t^2 \lambda^2}{N}}.
\end{align*}
\end{cor}

The symbol $\lesssim$ applies in the large-$N$ regime and suppresses constants.
The proof sketch is presented in Section~\ref{sec:proofideathm}, and the detailed proof is given in Section~\ref{sec:proofthm1}.

For comparison, the error bounds for the average quantum channel, established in \cite{campbell2019random}, imply that
{\small
\begin{align*}
&\max_{\ket{\psi}: \text{ state}} \,\,\, \max_{O: \norm{O}  \leq 1} \left| \mathbb{E}_{V_1, \ldots, V_N} [\langle O \rangle_{V_N \cdots V_1 \ket{\psi}}] -  \langle O \rangle_{\exp(-\mathrm{i} t H) \ket{\psi}} \right| 
\\
\lesssim & \frac{t^2 \lambda^2}{N}.
\end{align*}
}
As we can see, the error bound of the average over all product formulas is smaller than the error for an individual random product formula.
To understand the discrepancy, it is valuable to think about a randomly sampled product formula as a random walk on the group of $2^n \times 2^n$ unitary matrices.  Figure~\ref{fig:rwlie} indicates why a single realization of the random walk $V_N \cdots V_1$ has much greater error than the average of the random walks.

In the error bound $\mathcal{O}(\sqrt{nt^2 \lambda^2 / N})$, the square root reflects the statistical nature of the fluctuations in the random walk around its expectation.  The diamond norm requires us to control the behavior of the random product formula when applied to every $2^n$-dimensional input state.  Remarkably, we only pay for the \textit{logarithm} of the dimension, which coincides with the number $n$ of qubits.\AC{ We will discuss an intuitive example, where $\log(2^n)$ naturally arises from a union bound in Section~\ref{unionboundexp}. Formally, this pre-factor is a general feature of concentration inequality for matrix martingales (Sec.~\ref{sec:intro_martingale}).} Similar proof techniques could potentially be useful for controlling stochastic errors in other quantum computing applications.

\begin{figure}[t]
\centering
\includegraphics[width=0.9\columnwidth]{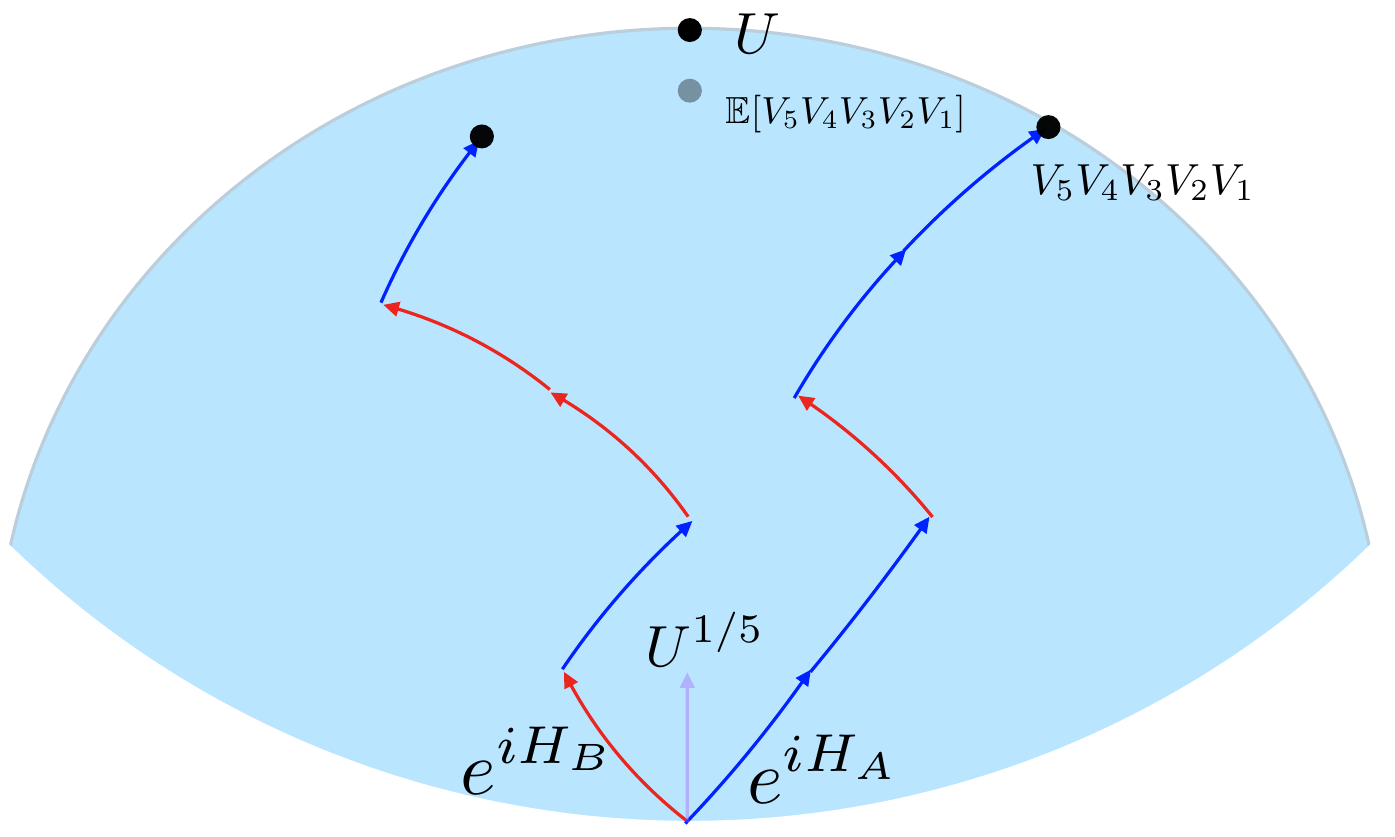}
\caption{
\textit{Illustration of concentration effects for random walks (and their averages) on the unitary group.} \\
The expectation $\mathbb{E}[ V_N \cdots V_1 ]$ of a random product formula is not unitary,
but it may be very close to the ideal evolution.
A sampled random product formula $V_N \cdots V_1$ is unitary,
but its distance from the ideal evolution is about $\mathcal{O}(\sqrt{n t^2 \lambda^2 / N})$. The average of the random product formulas results in an error of $\mathcal{O}\left(t^2 \lambda^2 / N\right)$.
}
\label{fig:rwlie}
\end{figure}

\subsection{Faster simulation for a fixed input state}
\label{taylortostate}

In practice, it is common to perform the quantum simulation starting from a particular input state.
The distinction with the previous setting is that the product formula only needs to work for one (arbitrary and possibly unknown) input state,
not for all states simultaneously.  The next theorem asserts that much shorter product formulas suffice in the easier setting.

\begin{thm}[qDRIFT: Gate complexity for fixed input] \label{thm:fixedinput}
 Consider an $n$-qubit Hamiltonian $H = \sum_j h_j$ with $\lambda = \sum_j \norm{h_j} $ and any input quantum state $\ket{\psi}$.
Draw a randomized product formula $V_N \cdots V_1$ from~\eqref{eq:qDRIFT} with the number of gates
 \begin{equation}
 \label{eq:singlerandomgate}
 N \geq \Omega(t^2 \lambda^2 \log(1 / \delta) / \epsilon^2).
 \end{equation}
 With probability at least $1 - \delta$, the output state $V_N \cdots V_1 \ket{\psi}$ satisfies
 \begin{equation*}
     \max_{O: O^\dagger = O, \norm{O}  \leq 1} \left| \langle O \rangle_{V_N \cdots V_1 \ket{\psi}} -  \langle O \rangle_{\exp(- \mathrm{i} t H) \ket{\psi}} \right| < \epsilon,
 \end{equation*}
 where $\langle O \rangle_{\ket{\psi}} = \bra{\psi} O \ket{\psi}$.
 This is equivalent to the output state $V_N \cdots V_1 \ket{\psi}$ being $\epsilon$-close to the ideal output state $\exp(-\mathrm{i} t H) \ket{\psi}$ in trace distance.
\end{thm}

Again, convexity allows us to extend this bound to a (fixed) mixed state $\rho = \sum_i p_i |\psi_i \rangle \! \langle \psi_i|$ as well. And, a complementary result bounds the expected approximation error in terms of gate count $N$.

\begin{cor}[qDRIFT: Error bound in trace distance]
\label{cor:exp_fixedinput}
Consider an $n$-qubit Hamiltonian $H = \sum_j h_j$ with $\lambda = \sum_j \norm{h_j}$.
A randomized product formula $V_N \cdots V_1$, drawn from \eqref{eq:qDRIFT}, 
has expected trace distance error for any fixed input
\begin{align*}
&\max_{\ket{\psi}: \text{ state}} \,\,\, \mathbb{E} \left[  \max_{O: \norm{O}  \leq 1} \left| \langle O \rangle_{V_N \cdots V_1 \ket{\psi}} -  \langle O \rangle_{\exp(-\mathrm{i} t H) \ket{\psi}} \right| \right] \\
\lesssim &\sqrt{\frac{ t^2 \lambda^2}{N}}.
\end{align*}
\end{cor}
 
Theorem~\ref{thm:fixedinput} yields an $n$-fold improvement over the gate count from Theorem~\ref{thm:allinput}.

So, for a 100-qubit system, focusing on a single input state leads to a $100\times$ reduction
in gate complexity over a simulation that works for all input states. The product formulas that work for a fixed input state may vary with the choice of state,
but they are all generated using the same \textsc{qDRIFT} procedure. 

\AC{The probabilistic origin of this improvement is in stark contrast with deterministic constructions, where
additional structure, such as low energy input states~\cite{low_energy2020}, is required to further reduce the gate count.
Here, we only make one assumption that is much weaker: the state mus be fixed and cannot depend on the concrete gate sequence being sampled.

We can even construct short product formulas that work for a moderate number of (arbitrary) input states by increasing the gate complexity slightly and invoking a union bound argument. 

The proof of Theorem~\ref{thm:fixedinput} is similar in spirit to the proof of Theorem~\ref{thm:allinput}. We construct a random walk from the (fixed) starting state $\ket{\psi}$.
We show that, with high probability, the output state is close to the ideal output state $U \ket{\psi}$. However, what marks the difference from the unitary results (Theorem~\ref{thm:allinput}) is that
\textit{vector} concentration inequalities suffice. 
More precisely, we analyze the vector random walk using a geometric tool, called uniform smoothness~\cite{HNTR20:Matrix-Product}. The details appear in Section~\ref{sec:proofthm2}.
The resulting $n$-fold improvement can also be understood as a result of switching the order of quantifiers, see~\ref{unionboundexp} for an explicit example. }

\subsection{
Extension to general products of random unitaries}
\AC{
So far, we have presented our results exclusively for \textsc{qDRIFT}. But the underlying concentration analysis readily extends to more general random walks on the unitary group. 
Let $V_N \cdots V_1 \in U(2^n)$ be a random product designed to approximate a target unitary $U=U_N \cdots U_1$. We need two properties. \newline
(i) \textit{Causality:} for $1 \leq k \leq N$ the random selection of $V_k$ can only depend on previous choices for $V_1,\ldots,V_{k-1}$:
\begin{align}
\mathrm{Pr} \left[ V_k |V_N\ldots V_{k+1} V_{k-1},\ldots V_1 \right]
=& \mathrm{Pr} \left[ V_k| V_{k-1},\ldots ,V_1 \right] \label{eq:causality}
\end{align}
(ii) \textit{accurate approximation:} Each realization of $V_k$ (and their conditional expectation) must be close to the ideal unitary $U_k$. More precisely, let $a_k,b_k >0$ be constants such that \begin{align}
\lV V_k-\BE_{k-1} V_k\rV\le a_k \quad \text{and} \quad 
\left\| \mathbb{E}_{k-1} V_k - U_k \right\| \leq b_k, \label{eq:small-steps}\\ \text{where}\ \ \ \ \mathbb{E}_{k-1}V_k:=\mathbb{E} \left[ V_k| V_{k-1},\ldots,V_1 \right] \nonumber
\end{align}
almost surely for all $1 \leq k\leq N$.
All concentration results from this work, as well as the main result from \cite{campbell2019random}, do readily extend to random products that do obey these two properties.

\begin{thm}[general concentration bounds for products of random unitaries]\label{thm:summary_of_errors}
Let $V=V_N \cdot V_1 \in U(2^n)$ be a random product that approximates a target product $U=U_N \cdots U_1$ in a \textit{causal} (Eq.~\eqref{eq:causality}) and \textit{small-step} (Eq.~\eqref{eq:small-steps} fashion.
Then, the associated unitary channels $\mathcal{V}(X) = VXV^\dagger$ and $\mathcal{U}(X) = UXU^\dagger$ obey

\begin{align*}
 \lV \CU-\CV\rV_\diamond &\le 2\sum_{k=1}^N a_k & \text{(worst case)},\\
 \mathbb{E} \lV \CU-\CV\rV_\diamond 
&\lesssim \sqrt{n\sum_{k=1}^N a_k^2}+ 2\sum_{k=1}^N b_k & \text{(typical case)},\\
\mathbb{E}\| \CU(\rho) - \CV(\rho) \|_1
&\lesssim \sqrt{\sum_{k=1}^N a_k^2}+ 2\sum_{k=1}^N b_k
& \text{(fixed input)},\\ 
\lV \CU-\mathbb{E}\CV\rV_\diamond &\leq 2 \sum_{k=1}^N b_k &
\text{(average case)}
\end{align*}
where $\lesssim$ suppressed absolute constants.
\end{thm}

\subsubsection{Concentration of randomly permuted Suzuki formulas}
So far, we have only introduced the first order Lie-Trotter formulas (\ref{eq:first_order}). The 2nd
order formula is typically called the Suzuki formula. For some short time $\tau >0$, define
\begin{align*}
    S_2(\tau) := \prod^L_{j=1}\exp(-\iunit(\tau/2) h_j ) \prod^1_{j=L}\exp(-\iunit(\tau/2) h_j ).
\end{align*}
Higher order formulas, also named after Suzuki, are constructed recursively from $S_2 (\tau)$:
\begin{equation*}
S_{2p}(\tau) := S_{2p-2}(q_p\tau)^2 S_{2p-2}((1-4q_p)\tau)S_{2p-2}(q_p\tau)^2,
\end{equation*}
where $q_p:=1/(4-4^{1/(2p-1)})$~\cite{suzuki1991general}. We can combine $r$ rounds of $2p$-th order Suzuki formulas with time $\tau = t/r$ each to approximate a desired quantum evolution. This yields
\begin{align*}
U = \e^{-\iunit t H} =& \e^{-\iunit (t/r) H} \cdots \e^{-\iunit (t/r) H} \\
\approx & S_{2p}(t/r) \cdots S_{2p}(t/r) = S_{2p}(t/r)^r
\end{align*}

For $\epsilon$ fixed, we choose an appropriate number of rounds $r$ and obtain a gate count $N\approx rL$ that scales as
\begin{equation}
N_{\mathrm{det}} \gtrsim t\Lambda L^2 \left(\frac{t\Lambda L}{\epsilon}\right)^{1/2p}
\quad \text{with} \quad \Lambda = \max_j \|h_j\|,\label{eq:old_gatecount}
\end{equation}
simple gates $\mathrm{e}^{-\iunit (t/r) h_j}$ to ensure worst-case approximation error $\| \mathcal{U}-\mathcal{V} \|_\diamond\le \epsilon$~\cite{childs2019faster}.
Note that the $2p$-th order Suzuki formula $S_{2p}(\tau)$ does not specify a preferred order in which the terms in Hamiltonian $h_j$ should appear.
Childs, Ostrander and Su observed that randomly permuting the order of Hamiltonian terms
within each block $S_{2p}(t/r)$ can suppress the approximation error for low order Suzuki formulas~\cite{childs2019faster} considerably. More precisely, 
we reshuffle the ordering in each Suzuki block by applying uniformly random permutations $\pi_1,\ldots,\pi_r$ to the term indices ($h_j \mapsto h_{\pi_1 (j)}$ for the first Suzuki block, etc.). Denote the resulting randomized Suzuki formula by $V_{2p}^{\pi_r\cdots \pi_1}= S_{2p}^{\pi_r}(t/r) \cdots S_{2p}^{\pi_1}(t/r)$.
Childs, Ostrander and Su proved that 
a total gate count

\begin{align}
    N_{\mathrm{avg}} \gtrsim &  t \Lambda L^2 \left(\frac{t\Lambda}{\epsilon}\right)^{1/2p} \label{eq:avg_suzuki}
\end{align}
ensures that the average channel (averaged over all possible permutations) obeys
$
\lV \CU-\mathbb{E}_{\pi_r,\ldots,\pi_1} \CV_{2p}^{\pi_r \cdots \pi_1}\rV_\diamond\le \epsilon.
$

We note in passing that Eq.~\eqref{eq:avg_suzuki} is tighter than the original bound presented in~\cite{childs2019faster}. This slight improvement follows from replacing the original mixing Lemma~\cite{campbell_mixing16,hastings2016turning} in the proof of Childs \textit{et al.} by a stronger statement derived in this work (Lemma~\ref{lem:mixing_maintext} below).

Applying Theorem~\ref{thm:summary_of_errors} to the 
problem at hand supplies strong concentration around this expected behavior.

\begin{cor}[concentration of randomly permuted Suzuki-formulas]\label{cor:randomsuzuki}
Consider a $n$-qubit Hamiltonian $H=\sum_j h_j$ with $\Lambda = \max_j \|h_j\|$, an associated time-$t$ evolution $U=\e^{-\iunit t H}$ and a Suzuki order $2p$.
Then, a total gate count of 
\begin{equation*}
    N_{\mathrm{typ}}  \gtrsim  t\Lambda L^2 \left(\frac{n\Lambda Lt}{\epsilon^2}\right)^{1/(4p+1)}+t \Lambda L^2 \left(\frac{t\Lambda}{\epsilon}\right)^{1/2p}, 
\end{equation*}
ensures that a randomly permuted Suzuki formula with $r$ terms obeys $\mathbb{E}_{\pi_1,\ldots,\pi_r} \| \mathcal{U} - \mathcal{V}_{2p}^{\pi_r \cdots \pi_1} \cdots S_{2p}^{\pi_1}(t/r) \| \leq \epsilon$ (typical case). For a fixed (but arbitrary) input state, the gate count can be further reduced to
\begin{equation*}
N_{\mathrm{fix}} \gtrsim  t\Lambda L^2\left(\frac{\Lambda Lt}{\epsilon^2}\right)^{1/(4p+1)}+t \Lambda L^2 \left(\frac{t\Lambda}{\epsilon}\right)^{1/2p}.
\end{equation*}

For simplicity, we omitted the logarithmic dependence on failure probability $\delta$.

\end{cor}

In contrast to qDRFIT, the parameters $n, \epsilon, L$ parameters now all raised to the $1/(4p+1)$st power, and the different randomized settings yield very much the same gate count $t\Lambda^2 L$ for large $p$. This is in accordance with observations provided  in~\cite{childs2019faster}.

For illustration, we have chosen to directly import older results (\ref{eq:old_gatecount}) to compare with the averaged channel bounds (\ref{eq:avg_suzuki}). However,  (\ref{eq:old_gatecount}) can possibly be sharpened to scale with commutator~\cite{childs2019theory}. Future work remains to carry out commutator analysis for the averaged channel(\ref{eq:avg_suzuki}) and then swiftly upgrade Corollary~\ref{cor:randomsuzuki}.

We refer to the appendix (Sec.~\ref{sec:proofgeneral}) for detailed statements and proofs. 

\subsubsection{Compiling}
Beyond Hamiltonian simulation, random product also help suppressing error in compiling. The task of compiling turns a perfect circuit diagram into a sequence of executable gates. However, since the gate set is discrete (or imperfect), the compiler can only return an approximate circuit 
\begin{align}
    \|U_N\cdots U_1 - V_N\cdots V_1\| \le \epsilon,
\end{align}
where each $V_k$ will be further decomposed into universal gates (such as using the Solovay-Kiteav Theorem, see e.g.~\cite{dawson2005kitaev}) up to some individual accuracy $\|U_k-V_k\|\le \eta$. In the worst case, the local error $\|U_k-V_k\|\le \eta$ accumulates linearly and we must conclude
\begin{align}
    \|U-V\| \le N\eta,  \label{eq:coherent-error}
\end{align} 
by means of the triangle inequality. This means that individual accuracy $\eta =\epsilon/N$ is required
to ensure an overall approximation error of (at most) $\epsilon$. 
This in turn, requires more gates for each approximation, because $\eta$-approximating $V_k$ requires a gate count\footnote{For the Solovay-Kiteav Theorem, $3\le c \le 4$~\cite{dawson2005kitaev} and it is also known that $c=1$ is the best possible exponent~\cite{Harrow_2002}} proportional to $\log^c(\eta)$.

Randomized compiling~\cite{hastings2016turning,campbell_mixing16} addresses precisely this issue. 
It uses random gate synthesis to avoid that individual approximation errors add up ``coherently'' over the entire compilation. The resulting ``incoherent'' error accumulation can be much more favorable and the mixing Lemma~\cite{hastings2016turning,campbell_mixing16} makes this intuition precise. In the following, we present a sharpened version of this statement that seems to be novel.

\begin{lemma}[Mixing lemma]\label{lem:mixing_maintext}
Let $U_k$ be a fixed unitary and $V_k$ a random unitary. Then,
\begin{equation*}
     \lV  \mathcal{U}_k-\BE \mathcal{V}_k \rV_{\diamond} \le 4 \lV U_k- \BE V_k\rV.
\end{equation*}
\end{lemma}
\noindent We refer to Appendix~\ref{sub:conversion} for a self-contained proof. 

Back to compiling, $\lV \BE V_k-U_k\rV $ can as small as $\CO(\eta^2)$~\cite{campbell_mixing16} by mixing appropriate synthesis protocols for $V_k$. This yields an improved bound for the average channel,
\begin{align}
    \lV \mathcal{U}-\BE \mathcal{V} \rV_{\diamond} \le \CO(N\eta^2),\label{eq:random_compiling}
\end{align}
In words, the individual accuracy needed is suppressed quadratically to $\eta = \Omega(\sqrt{\epsilon/N})$.
In~\cite{campbell_mixing16}, Campbell pointed out that this square root improvement may turn into a multiplicative overhead reduction, depending on the gate synthesis efficiency $\log^c(\sqrt{\eta})=\log^c(\eta)/2^c $.
Using Theorem~\ref{thm:summary_of_errors}, we can immediately bound the performance of randomized compiling \textit{without mixing}.

\begin{cor}[Randomized compiling without mixing]
Suppose that we wish to approximate $U=U_N \cdots U_1$ by a gate collection $V=V_N \cdots V_1$ such that each $V_k$ is random and obeys $\|U_k-V_k\|\le \eta$ almost surely. Then,
\begin{equation*}
\BE  \| U - V \| \lesssim \sqrt{nN} \eta.
\end{equation*}
What is more, the $\sqrt{n}$-factor on the r.h.s.\ disappears if we restrict attention to a fixed input state.

\end{cor}
This translates to individual accuracy $\eta = \CO(\epsilon/\sqrt{nN})$, and $\eta = \CO(\epsilon/\sqrt{N})$ respectively. Both results interpolate between the worst case~\eqref{eq:coherent-error} (``coherent'' error accumulation) and best case~\eqref{eq:random_compiling} (``incoherent'' error accumulation).

}

\subsection{Numerical experiments} \label{sec:numerics}

\begin{figure*}[t]
\centering
\includegraphics[width=1.0\textwidth]{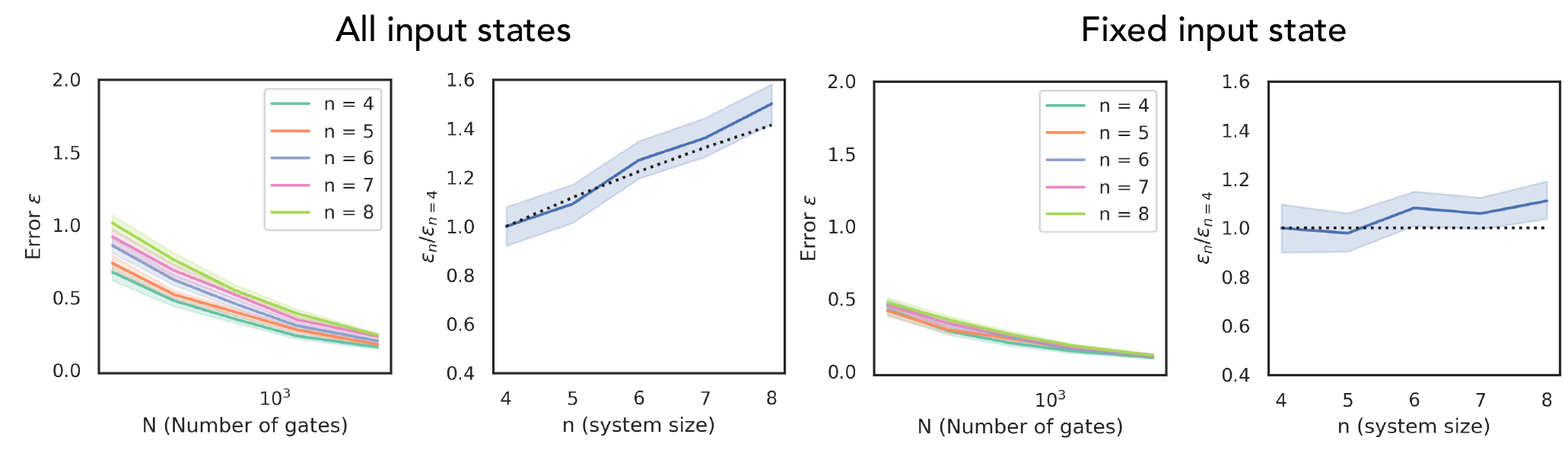}
\caption{\AC{
\textit{Numerical experiments for simulating 1D Heisenberg model under different gate count $N$.} \\
In \textit{All input states} (left), we consider $\epsilon = \norm{U - V_N \ldots V_1}_\infty$, which considers the error over all input states and observables. In \textit{Fixed input state} (right), we consider $\epsilon = \norm{U \ket{\psi} - V_N \ldots V_1 \ket{\psi}}_{\ell_2}$, which considers the error over all observables. The input state $\ket{\psi}$ is chosen to be the tensor product of single-qubit Haar-random states.
For both \textit{All input state} and \textit{Fixed input state}, we give an additional plot showing how the error $\epsilon$ increases as the system size $n$ increases for a fixed number of gates $N = 160$.
The $y$-axis is normalized using the average error for system size $n = 4$ over $50$ independent runs.
Bounds in Theorem~\ref{thm:allinput}~and~\ref{thm:fixedinput} show that the relative error $\epsilon_n / \epsilon_{n=4}$ scales as $\sqrt{n / 4}$ for \textit{All input state} and stays as constant $1$ for \textit{Fixed input state}, which are shown as the dotted lines.
The shaded regions are the standard deviation over $50$ independent runs.}
}
\label{fig:exper}
\end{figure*}
\AC{
In this section, we perform numerical experiments for simulating a simple Heisenberg model on a one-dimensional chain with a randomly sampled product formula. 
For $n$ qubits, $H = \frac{1}{n-1} \sum_{i=1}^{n-1} X_i X_{i+1} + Y_i Y_{i+1} + Z_i Z_{i+1}$ and we view this as a sum of $3(n-1)$ simple terms.
The normalization keeps the interaction strength $\lambda =3$ as a constant and we consider a constant time evolution $t=2$.The numerical experiments for the error under various setups using different gate count $N$ and different system sizes $n$ are shown in Figure~\ref{fig:exper}.
We can see that the error $\epsilon$ when we consider all input states scales as $\sqrt{n}$.
In contrast, the error $\epsilon$ stays roughly the same when we only consider a single input state.
This is in accordance with the theoretical predictions presented in Theorem~\ref{thm:allinput}~and~\ref{thm:fixedinput}.
}

\section{Instructive examples and proof idea}
\label{sec:proofidea}

\subsection{Comparison between stochastic averages of product formulas and concrete instances}
\label{nvs1}

This section considers an extremely simple Hamiltonian to pinpoint important differences between averaging random product formulas (that is, Campbell's black box) and concrete instances of product formulas.
The example Hamiltonian is a 1-local non-interacting Hamiltonian with a Pauli-$Z$ operator acting on each qubit:
\begin{align}
H =& \frac{1}{n}\sum_{k=1}^n Z_k,
\quad \text{where} \label{eq:example-hamiltonian}\\
Z_k =& \underset{\text{$(k-1)$-times}}{\underbrace{\mathbb{I} \otimes \cdots \otimes \mathbb{I}}} \otimes Z \otimes 
\underset{\text{$(n-k)$-times}}{\underbrace{\mathbb{I} \otimes \cdots \otimes \mathbb{I}}} 
\nonumber
\end{align}
for $1 \leq k \leq n$.
The relevant parameters are $L=n$ (number of terms), $\lambda =\tfrac{1}{n}\sum_{k=1}^n \|Z_k \|  = 1$ (interaction strength) and we fix the evolution time to $t=\pi$. 

Stochastic averages of random product formulas can accurately approximate the associated unitary evolution $U=\exp(- \iunit \pi H)$ after only a few iterations. 
The following observation is an immediate consequence of Campbell's main result \cite{campbell2019random}, see also Proposition~\ref{prop:bias1} below.

\begin{cor} \label{obs:example}
Fix a target accuracy $\epsilon$ and set $N = t^2 \lambda^2/\epsilon \approx 10/\epsilon$.Then, $N$ successive applications of the \textsc{qDRIFT} single-step average $\mathcal{V}(\rho)=\tfrac{1}{n}\sum_k \exp(- \mathrm{i} \tfrac{\pi}{N} Z_k ) \otimes \mathbb{I}^{(\text{else})} \rho \exp( \mathrm{i} \tfrac{\pi}{N} Z_k) \otimes \mathbb{I}^{(\text{else})}$ (Campbell's black box)
approximate the target unitary channel $\mathcal{U}(\rho)=U\rho U^\dagger$ up to accuracy $\epsilon$ in diamond distance. In particular, $\tfrac{1}{2}\| \mathcal{V}^{(N)}(|\psi \rangle \! \langle \psi|) - U |\psi \rangle \! \langle \psi| U^\dagger \|_1 \leq \epsilon$ for all input states $|\psi \rangle \! \langle \psi|$.
\end{cor}

This assertion seems remarkably strong. In particular, the sequence length $N$ does not depend on the number of qubits $n$. Once $n$ is sufficiently large it becomes impossible for concrete product formulas to achieve comparable results. The problem is that the sequence length $N$ is too small to address all $n$ qubits. This necessarily leads to substantial discrepancies between the simulated time evolution $V_N \cdots V_1$ and the actual target $U$, see Figure~\ref{fig:determin-lower} for an illustration.

\begin{lemma}
\label{lem:example-negative}
Assume that $n$ is an even number.
It is impossible to accurately approximate the time evolution $U$ defined in Eq.~\eqref{eq:example-hamiltonian} with  $N < n/2$ elementary gates of the form $V_i=\exp (-\iunit \tfrac{\pi}{N}Z_{k(i)}) \otimes \mathbb{I}^{(\text{else})}$.
More precisely, for each product formula $V=V_N \cdots V_1$, there exists an input state $|\psi \rangle \! \langle \psi|$ such that $\tfrac{1}{2}\|V |\psi \rangle \! \langle \psi| V^\dagger - U|\psi \rangle \! \langle \psi| U^\dagger \|_1 =1$.
\end{lemma}

\begin{figure}[t]
    \centering
    \includegraphics[width=1\columnwidth]{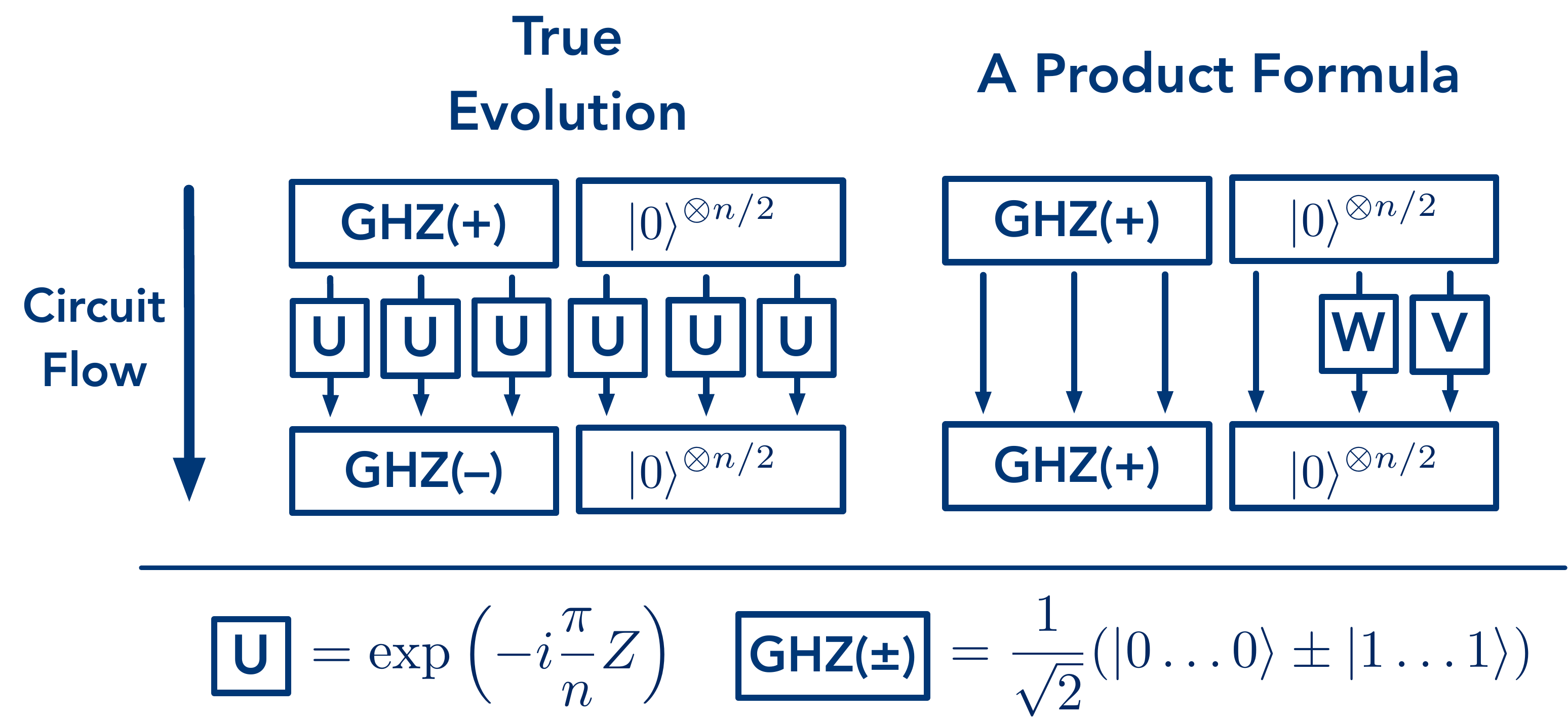}
    \caption{\textit{Illustration of the worst-case input for a product formula simulating evolution of a simple Hamiltonian.} \\
    The Hamiltonian $H = \frac{1}{n} \sum_{k=1}^n Z_k$ produces a time evolution that factorizes into single qubit unitaries $U$ (left). A product formula with fewer than $n/2$ single-site terms (right) is too small to address all qubits; at least $n/2$ of them must remain untouched.
    These errors accumulate for a GHZ-state comprised of these untouched qubits. If $n$ is large, even small evolution times ($U = \exp (-\mathrm{i} \tfrac{\pi}{n}Z)$) can accumulate and lead to a maximal approximation error ($\langle \mathrm{GHZ}(+), \mathrm{GHZ}(-) \rangle=0$).
   }
    \label{fig:determin-lower}
\end{figure}

\begin{proof}
All terms in the Hamiltonian \eqref{eq:example-hamiltonian} commute. Hence, the associated target evolution factorizes nicely into tensor products: $U = \exp (-\iunit  \pi H) = \exp (-\iunit  \tfrac{\pi}{n}Z_1) \otimes \cdots \otimes \exp (-\iunit  \tfrac{\pi}{n}Z_n)$. Up to a global phase, each single-qubit unitary affects the computational basis in the following fashion: $\exp( -\iunit  \tfrac{\pi}{n}Z) |0 \rangle = |0 \rangle$ and $\exp (-\iunit  \tfrac{\pi}{n} Z) |1 \rangle = \exp (\iunit \tfrac{2\pi}{n}) |1 \rangle$. 
These small phase shifts can add up for states that are in superposition. Consider the tensor product of a GHZ state on $n/2$ qubits with the all-zeroes state on the remaining half: $|\tilde{\psi}\rangle = \tfrac{1}{\sqrt{2}}(|0 \rangle^{\otimes n/2} + |1 \rangle^{\otimes n/2}) \otimes |0 \rangle^{\otimes n/2}$.
Then,
\begin{align*}
U |\tilde{\psi} \rangle =& \exp(-\iunit  \tfrac{2\pi}{n} Z)^{\otimes n}|\tilde{\psi}\rangle \\
=& \tfrac{1}{\sqrt{2}}( |0 \rangle + (\exp (\iunit \tfrac{2\pi}{n}))^{n/2} |1 \rangle) \otimes |0 \rangle^{\otimes n/2} \\
=& \tfrac{1}{\sqrt{2}}(|0 \rangle^{\otimes n/2}- |1 \rangle^{\otimes n/2}) \otimes |0 \rangle^{\otimes n/2}
\end{align*}
and we can easily check that input and output are orthogonal to each other: $\tfrac{1}{2}\|U |\tilde{\psi} \rangle \! \langle \tilde{\psi}| U^\dagger - |\tilde{\psi} \rangle \! \langle \tilde{\psi}| \|_1 = 1$.
These features do not change if we permute the qubits in the input state $|\tilde{\psi} \rangle$. Any combination of a GHZ state on one half of the qubits with computational $|0\rangle$-states on the remaining ones obeys the same orthogonality relation. We can use this freedom to construct a worst-case input $|\psi \rangle$ for a fixed product formula $V=V_N \cdots V_1$ comprised of fewer than $n/2$ single-qubit gates. Simply initialize the (at most) $n/2$ qubits on which the product formula acts nontrivially in the computational 0-state and hide the GHZ component among the remaining qubits. By construction, the product formula $V$ does not affect this input state at all.
This is a worst case, because the target unitary $U$ does rotate the hidden GHZ component into an orthogonal configuration:
$\| U |\psi \rangle \! \langle \psi|U^\dagger - V|\psi \rangle \! \langle \psi| V^\dagger \|_1 = \| U |\psi \rangle \! \langle \psi| U^\dagger - |\psi \rangle \! \langle \psi| \|_1 =1$.
\end{proof}

This negative statement highlights that the gate count of (worst case) accurate product formulas must in general depend on the number of qubits and justifies the appearance of $n$ in Theorem~\ref{thm:allinput}. Note, however, that Lemma~\ref{lem:example-negative} is contingent on identifying a worst-case input state for a fixed (and known) product formula. If the input state is fixed, the situation can change dramatically. For instance, we could use explicit knowledge of the input to construct a product formula that accurately approximates its time evolution. Identifying an optimal product formula seems like a daunting task, but randomness can help. 
Theorem~\ref{thm:fixedinput} asserts that a collection of $N \gtrsim \pi^2/\epsilon^2$ randomly selected single-qubit gates approximate the time evolution \eqref{eq:example-hamiltonian} of any fixed input state $|\psi \rangle$ up to accuracy $\epsilon$ in trace distance. While this gate count is considerably larger than the one put forth in Corollary~\ref{obs:example}, it is still independent of the number of qubits. What is more, this assertion applies with high probability to \textit{any} fixed input state. 
This capitalizes on another advantage of generating unstructured product formulas according to a randomized procedure: it is extremely difficult to fool a randomized compiling procedure with an already fixed 
 input. 
\subsection{Instructive concentration argument for a simple Hamiltonian}
\label{unionboundexp}

 \begin{figure}[t]
  \centering
  \includegraphics[width=\columnwidth]{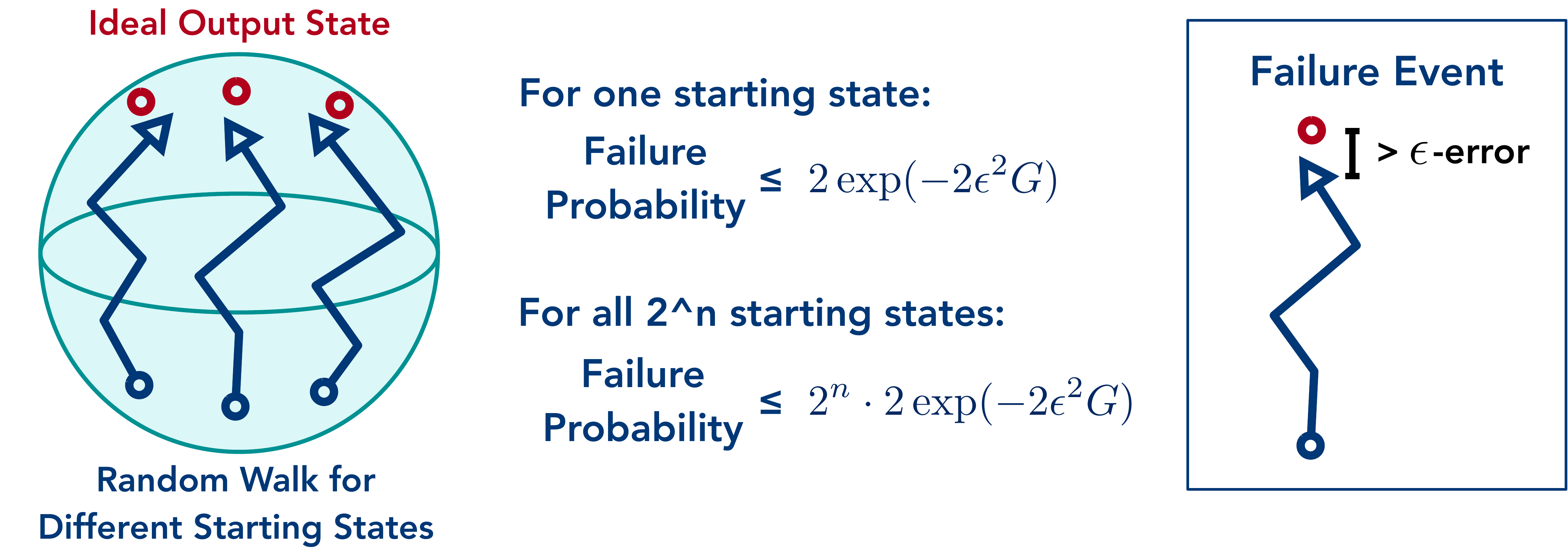}
  \caption{
  \textit{Illustration of the probabilistic proof for the commuting Hamiltonian given in Equation~\eqref{eq:CHZs}.}\\
  We consider all the $2^n$ computational basis states as the starting state. The probability for one of the starting state to incur at least an error $\epsilon$ is exponentially smaller than the probability for the maximum of the $2^n$ starting states to incur error $> \epsilon$. However the failure probability is exponential suppressed by increasing the gate count $N$. Hence one only need to set $N = n / \epsilon^2$.
  }
  \label{fig:factforCHS}
\end{figure}

This section provides intuition for the concentration effects that ultimately imply Theorem~\ref{thm:allinput}
by means of another example Hamiltonian that is composed of (commuting) Pauli-Z terms only: 

\begin{align}
\label{eq:CHZs}
    H =& \frac{1}{2^n} \sum_{\vct{p}\in \{0, 1\}^n} \alpha_\vct{p} Z_\vct{p}
    \quad \text{where} \\ Z_{\vct{p}}=&Z_{(p_1,\ldots,p_n)}=Z^{p_1}\otimes \cdots \otimes Z^{p_n} \nonumber
\end{align}
(with the convention that $Z^0 =\mathbb{I}$) and $\alpha_{\vct{p}} \in \left\{-1,1\right\}$.
That is, the Hamiltonian is a sum of $2^n$ signed Pauli strings that are comprised of $Z$ and $\mathbb{I}$, as well as a global sign. A high-order Suzuki formula would require a gate complexity of $\mathcal{O}(L) = \mathcal{O}(2^n)$ by putting down every term. In contrast, by random sampling, Theorem~\ref{thm:allinput} yields a gate complexity of $\mathcal{O}(n / \epsilon^2)$ (Theorem~\ref{thm:fixedinput} yields $\CO(1/\epsilon^2)$). This is an exponential improvement in system size.

\AC{
The physical intuition is that all the terms in the Hamiltonian act on the same system with $n$ qubits (a $2^n$-dimensional Hilbert space), so their actions must ``overlap'' with one another, and can be efficiently estimated by random sampling. 
To see this effect more clearly, let us write down the unitary evolution $\exp(-\mathrm{i} H)$ in the computational basis $\ket{\vct{b}}$ with multi-index $\vct{b}=(b_1,\ldots,b_n) \in \{0, 1\}^n$.
Note that all terms in the Hamiltonian \eqref{eq:CHZs} are diagonal in the computational basis. This implies
\begin{align}
\exp(-\mathrm{i} H) \ket{\vct{b}} =& \exp\left(-\mathrm{i} \frac{1}{2^n} \sum_{\vct{p} \in \{0, 1\}^n} \alpha_\vct{p} Z_\vct{p} \right) \ket{\vct{b}} \\
=& \exp\left(-\mathrm{i} \frac{1}{2^n} \sum_{\vct{p} \in \{0, 1\}^n} c_\vct{p}(\vct{b})\right) \ket{\vct{b}} \nonumber\\
:=& \mathrm{e}^{-\mathrm{i} S(\vct{b})} \ket{\vct{b}}, \nonumber
\end{align}
where $c_\vct{p}(\vct{b}) = \alpha_\vct{p}  \bra{\vct{b}} Z_\vct{p} \ket{\vct{b}} \in \{-1, 1\}$. When we instead sub-sample an effective Hamiltonian $\hat{H}$ comprised of $N$ randomly selected terms, the constructed product formula would be
\begin{align}
   \exp(-\mathrm{i} \hat{H})&= \exp\left(-\mathrm{i} \frac{1}{N} \alpha_{\vct{p}_N} Z_{\vct{p}_N}\right) \cdots \exp\left(-\mathrm{i} \frac{1}{N} \alpha_{\vct{p}_1} Z_{\vct{p}_1}\right) \ket{\vct{b}} \nonumber \\
    =& \exp\left(-\mathrm{i} \frac{1}{N} \sum_{k} c_{\vct{p}_k}(b)\right) \ket{\vct{b}} := \mathrm{e}^{-\mathrm{i} \hat{S}(\vct{b})} \ket{\vct{b}}. \nonumber
\end{align}
By Hoeffding's inequality, $\hat{S}(\vct{b}) = \frac{1}{N} \sum_{k} c_{\vct{p}_k}(\vct{b})$ should concentrate around $S(\vct{b}) = 2^{-n} \sum_{\vct{p} \in \{0, 1\}^n} c_\vct{p}(\vct{b})$ with standard deviation $1 / \sqrt{N}$ and an exponentially decaying tail (think Monte Carlo). An illustration and some helpful facts can be found in Figure~\ref{fig:factforCHS}. 

Let us now fix an (arbitrary) $n$-qubit basis state. Then, the probability of an $\epsilon$-deviation (or larger), i.e.\ \ $|\hat{S}(b) - S(b)| > \epsilon$, 
can be bounded by $1/\mathrm{e}$ if we choose $N = 1 / \epsilon^2$. However, if we wish the effective Hamiltonian to work \textit{for all} basis states simultaneously, we would choose a larger $N = n/ \epsilon^2$ to ensure that the deviation probability is further suppressed exponentially to $1 / \mathrm{e}^n$. By a union bound, $|\hat{S}(b) - S(b)| \leq \epsilon$ for all $2^n$ computational basis states $\ket{\vct{b}}$ with probability at least $1 - 2^n / \mathrm{e}^n$. Albeit in the simplest example (commuting Hamiltonian), this demonstrates that a random product formula $\exp(-\mathrm{i} \hat{H})$ can accurately simulate $\exp(-\mathrm{i} H)$ up to error $\epsilon$ with only $N = n / \epsilon^2$ gates, which further reduces to $\CO(1/\epsilon^2)$ for an arbitrary basis state. The powerful tool of concentration for matrix (vector) martingales allows us to prove the same statement for any (non-commuting) many-body Hamiltonian.}

We will return to this example Hamiltonian in Section~\ref{sec:tightness} to show that this more general analysis yields an essentially optimal parameter dependence:
dimension dependence that is tight: the scaling $N \geq \Omega( n t^2 \lambda^2/\epsilon^2)$ in Theorem~\ref{thm:allinput} is unavoidable in general.

\subsection{Proof idea for Theorem~\ref{thm:allinput}~and~\ref{thm:fixedinput}}
\label{sec:proofideathm}

This section sketches the main ideas and tools required to establish Theorem~\ref{thm:allinput}~and Theorem~\ref{thm:fixedinput}. The other results follow from more elementary
arguments. Detailed arguments and rigorous statements are provided in Appendix~\ref{sec:proofs} below.

Consider an $n$-qubit Hamiltonian $H = \sum_{j=1}^L h_j$ and an evolution time $t$.  The associated unitary evolution defines a (unitary) channel on $n$-qubit states:
\begin{align*}
\mathcal{U}(| \psi \rangle \! \langle \psi| ) =& U | \psi \rangle \! \langle \psi| U^\dagger \quad \text{where} \\
U =& \exp \left( -\mathrm{i} t H \right)
= \exp \Big( - \mathrm{i} t \sum\nolimits_{j=1}^L h_j \Big).
\end{align*}
Fix a number of steps $N$ and set $\lambda = \sum_{j=1}^L \|h_j \|$. The task is to accurately approximate the target unitary $U$ by a product formula, i.e., the composition of $N$ \textit{simple} unitary evolutions:
\begin{align*}
\mathcal{V}^{(N)} (|\psi \rangle \! \langle \psi|)  =& \mathcal{V}_N\circ \cdots \circ \mathcal{V}_1 (| \psi \rangle \! \langle \psi|) \\
=& V_N \cdots V_1 | \psi \rangle \! \langle \psi| V_1^\dagger \cdots V_N^\dagger. 
\end{align*}

We quantify the difference between $\CV^{(N)}$ and $\mathcal{U}$ in \textit{diamond distance}. That is, the worst case approximation error over all possible input states $\rho$ in the presence of an unaffected quantum memory.

Let $\mathcal{E},\mathcal{F}$ be two quantum channels, and let $\mathcal{I}(|\psi \rangle \! \langle \psi|) = |\psi \rangle \! \langle \psi|$ denote the identity channel on an equally large ancilla system. The diamond distance between $\mathcal{E}$ and $\mathcal{F}$ is defined as
\begin{equation}
\tfrac{1}{2}\| \mathcal{E}-\mathcal{F} \|_\diamond = \tfrac{1}{2}\max_{|\psi \rangle \! \langle \psi|} \| \mathcal{E} \otimes \mathcal{I}(|\psi \rangle \! \langle \psi|) - \mathcal{F} \otimes \mathcal{I} (|\psi \rangle \! \langle \psi|) \|_1,
\label{eq:stabilized-diamond-distance}
\end{equation}
where the maximization ranges over all pure\footnote{Convexity ensures that the worst-case discrepancy is attained at a pure state $|\psi \rangle \! \langle \psi|$. Hence, it is not necessary to consider mixed states $\rho$ in this definition. We refer to \cite{watrous2018book} for details.}
input states $| \psi \rangle \! \langle \psi|$ and
$\| \cdot \|_1$ denotes the trace norm.
First, we relate the diamond distance \eqref{eq:stabilized-diamond-distance} between the channels $\CV^{(N)}$ and $\mathcal{U}$, which are both unitary, to an operator norm distance of the associated matrices:
\begin{align}
 \tfrac{1}{2}\| \mathcal{V}^{(N)}-\mathcal{U} \|_\diamond  
=&\tfrac{1}{2} \max_{|\psi \rangle \! \langle \psi|} \| \mathcal{U}(|\psi \rangle \! \langle \psi|) - \mathcal{V}^{(N)} (|\psi \rangle \! \langle \psi|) \|_1 
\nonumber \\
\leq & \| V_N \cdots V_1 - U \|,
\label{eq:diamond-distance}
\end{align}
see Lemma~~\ref{lem:diamond-to-operator} below.
This relation exploits the fact that stabilization
 (i.e., tensoring with the identity channel) is not
necessary for computing the diamond distance of two unitary channels~\cite[Thm.~3.55]{watrous2018book}.

Now, we can deal with the i.i.d.~random matrices $V_N,\ldots,V_1$ in the more familiar operator norm. Add and subtract the expected product $\mathbb{E} \left[ V_N \cdots V_1 \right] = \mathbb{E} \left[V_N \right] \cdots \mathbb{E} \left[ V_1 \right]= \left( \mathbb{E}V \right)^N$ to decompose the operator-norm difference into two qualitatively different contributions:
\begin{align}
& \| V_N \cdots V_1 - U \|  \nonumber \\
\leq&  \underset{\text{deterministic bias}}{\underbrace{ \| (\mathbb{E} V )^N - U \| }}
+
\underset{\text{random fluctuation}}{\underbrace{\| V_N \cdots V_1 - \mathbb{E} \left[ V_N \cdots V_1 \right] \| }}.
\label{eq:reformulation}
\end{align}
These two contributions can be analyzed separately:\AC{
\begin{enumerate}
\item[i.] \textit{Deterministic bias:} Most product formulas arise from first decomposing the target unitary into a sequence of many small steps: $U=(U^{1/N})^N$, where $U^{1/N} = \exp \left( - \mathrm{i} (t/N) H \right)$ is close to the identity matrix. 
This allows for approximating $U^{1/N}$ by another process that is easier to implement.
The random importance sampling model \eqref{eq:qDRIFT} over individual Hamiltonian terms is designed to achieve this goal. The average approximation error scales inverse quadratically in the number of steps:
$
\| (\mathbb{E} V ) - U^{1/N}\| \leq t^2 \lambda^2 / N^2
$;
see Lemma~\ref{lem:taylor} below.
While small, this expected error does constitute a bias that is present in each of the $N$ approximation steps. It can, and in general will, accumulate across different time steps:
\begin{align}
\| \mathbb{E} \left[V_N \cdots V_1 \right] - U \|  =& \| (\mathbb{E} V )^N - (U^{1/N})^N \| \nonumber \\
\leq & N \| (\mathbb{E}V)-U^{1/N}\| \nonumber  \\
\leq & \frac{t^2 \lambda^2}{N},
\label{eq:campbell2}
\end{align}
see Lemma~\ref{lem:composition} below.
The first inequality is obtained from a telescoping sum. This upper bound diminishes as the number of steps $N$ increases. For $\epsilon >0$,
\begin{equation}
N \geq \frac{2 t^2 \lambda^2}{\epsilon}
\quad \text{ensures} \quad \|(\mathbb{E}V)^N-U \|
\leq \frac{\epsilon}{2}.
\label{eq:bias}
\end{equation}

\item[ii.] \textit{Random fluctuation:} 

We also need to control the deviation of a product of i.i.d.~unitaries $V_N \cdots V_1$ from its expectation $\mathbb{E}[V_N \cdots V_1]= (\mathbb{E} V)^N$ in operator norm. We achieve this by applying modern matrix martingale tools, which may of independent interest. In short (see ~Sec.\ref{sec:intro_martingale} for a more thorough introduction), a \textit{martingale} is a random process $\left\{B_k:\; k=0,\ldots,N\right\}$ such that the distribution of $B_k$ only depends on past elements $B_{k-1},\ldots,B_1$ (`causality') and also $\mathbb{E}[B_k|B_{k-1}\cdots B_1]=B_{k-1}$ (`on average, tomorrow is the same as today').
To control fluctuations, we introduce a martingale that interpolates between the extreme cases we need to compare:
\begin{align*}
B_k =& (\mathbb{E}V)^{N-k}V_k \cdots V_1 \quad \text{such that} \\
B_0 =& (\mathbb{E}V)^N \quad \text{and} \quad 
 B_N= V_N \cdots V_1.
\end{align*}
Note that adjacent elements of this process only differ in a single term: $B_k$ arises from $B_{k-1}$ by replacing $\mathbb{E}V$ at position $k$ (counted from the right) by a random realization $V_k$ of $V$,{ and thus $\BE_{k-1} B_k=B_{k-1}$}.  Moreover, the current iterate $B_k$ only depends on realizations in the past and the interpolation process $\left\{B_k:\;k=0,\ldots,N\right\}$ forms a martingale comprised of $d \times d$ matrices, where $d=2^n$ is the Hilbert space dimension. The discrepancies are small: $\|V_k - (\mathbb{E} V)\| \leq 2t \lambda/N$, because each realization of $V$ is also very close to the identity matrix. 

Powerful tail bounds for matrix-valued martingales (which are relatively modern compared to their scalar counterparts) are available in the literature~\cite{oliveira2009spectrum,tropp2011freedman}. 
Adapting these results to the task at hand yields the bound
\begin{align}
&\mathrm{Pr} \left[ \| V_N \cdots V_1 - (\mathbb{E}V)^N \|  \geq \epsilon/2 \right] \\
&\leq 2d \exp \left( - \frac{N\epsilon^2}{44t^2 \lambda^2} \right), \label{eq:deviation-bound}
\end{align}
see Proposition~\ref{prop:fluctuation} below.
In words, the product $V_N \cdots V_1$ will concentrate around its expectation once $N$ is sufficiently large. 
Similar to more conventional random walks on integer lattices, the error is subgaussian with variance proportional to $ N \cdot (\lambda^2 t^2 / N^2)$.  There is an extra dimensional factor $d=2^n$ that arises because the martingale is matrix-valued. It converts to the number of qubits $n=\log(d)$ in the gate count $N$. For error parameters $\epsilon,\delta \in (0,1)$,
a gate count obeying
\begin{align}
& N \geq 44 \frac{ t^2 \lambda^2}{\epsilon^2} \log (2d/\delta) \quad \text{implies} \label{eq:random-fluctuation} \\
&\| V_N \cdots V_1 - (\mathbb{E} V)^N \|  \leq \tfrac{\epsilon}{2}
\text{ with probability $> 1-\delta$.}
\nonumber
\end{align}
\end{enumerate}
}\AC{
Theorem~\ref{thm:allinput} can be derived by combining the bound \eqref{eq:bias} for the deterministic bias with the bound \eqref{eq:random-fluctuation}  for random fluctuations. This provides a high probability error bound in the diamond distance when $N \geq \Omega(n t^2 \lambda^2 / \epsilon^2)$. 

Corollary~\ref{cor:exp_allinput} is derived by integrating the tail bounds of Theorem~\ref{thm:allinput}.
This produces

\begin{equation*}
\mathbb{E} \|V_N \cdots V_1 - (\mathbb{E} V)^N \|  \lesssim \sqrt{\frac{t^2 \lambda^2}{N} \log_2 (d)},
\end{equation*}
where the symbol $\lesssim$ suppresses a modest multiplicative constant. We refer to Section~\ref{sub:expectation} for details.\\
With fixed input(Theorem~\ref{thm:fixedinput}), we
use convexity to restrict our attention to a fixed, pure input state $|\psi \rangle$. 
We then need to compare $U|\psi \rangle$ with $V_N \cdots V_1 |\psi \rangle$ which can be achieved by constructing an interpolating random process that describes a \textit{vector}-valued martingale. 
We then call another concentration inequality (Lemma~\ref{lem:fixed-input-fluctuation})
\begin{equation*}
    \mathrm{Pr}\big[\| (V_N \cdots V_1 - \mathbb{E}[V]^N)| \psi \rangle \|_{\ell_2}
    > \epsilon \big]
    \leq \exp\left( \frac{-\epsilon^2 N}{8\mathrm{e} t^2 \lambda^2} \right),
\end{equation*}
which, in contrast, does not contain a dimensional factor. This is the reason why Theorem~\ref{thm:fixedinput} and Corollary~\ref{cor:exp_fixedinput} do not depend on system size at all. 
}

\section{Discussion and outlook}
\AC{

This work shows interesting characteristics of randomization that might help to further improve quantum simulation. \textit{(a)}
By studying typical unitary instances of \textsc{qDRFIT}, we have shown that $L$-independence (the number of terms in Hamiltonian) of \textsc{qDRIFT} attributes to randomly sampling terms; mixing different realizations is not essential. 
\textit{(b)} 
Gate complexities can be reduced substantially by restricting attention to a particular input state or/and target observable. 
Simple randomized compilation procedures, like \textsc{qDRIFT}, do not utilize extra information. We believe that more specialized algorithms might be able to exploit additional structure. 
\textit{(c)} We have shown that channel averages can be much closer to the ideal evolution than any individual product formula, however, in the case of \textsc{qDRIFT} it is only a quadratic improvement $1/\epsilon^2\rightarrow 1/\epsilon$. This can be seen as a manifestation of Jensen's inequality for convex functions (distance to ideal evolution) and we may also see such behavior in other quantum simulation algorithms.
Up to our knowledge, we also provide the first matrix concentration analysis for a randomly sampled product formulas and -- more generally -- Hamiltonian simulation. Similar proof techniques readily apply to any random product formula, e.g. randomly permuted Suzuki formulas~\cite{childs2019faster}, and in fact any (causal) random unitary product, e.g. randomized compiling~\cite{hastings2016turning,campbell_mixing16}.

Beyond random products, we expect the developed matrix concentration tools to be useful for controlling stochastic errors in other quantum computing applications, as well as analyzing properties of random Hamiltonians~\cite{chen2021concentration}.}
\subsection*{Acknowledgments} 
The authors want to thank John Preskill and Yuan Su for their valuable input and inspiring discussions. Earl Campbell and Nathan Wiebe provided insightful comments and encouraging feedback. We would especially like to thank Richard Allen and Angus Lowe for pointing out a significant gap in the proof of Lemma 3.4 and Richard Allen and Haimeng Zhao for proposing an elegant proof that fixed the claim with only an additional factor of 2. We are also grateful to Soonwon Choi for relaying this message. Finally, we would like to thank the anonymous reviewers for their in-depth comments and suggestions.
CC is thankful for Physics TA Relief Fellowship at Caltech.
HH is supported by the Kortschak Scholars Program.
RK acknowledges funding from ONR Award N00014-17-1-2146 and ARO Award W911NF121054).
JAT gratefully acknowledges funding from the ONR Awards N00014-17-1-2146 and N00014-18-1-2363 and from NSF Award 1952777.

\onecolumngrid\par

\appendix 
{
\section{Calculations for Table~\ref{table:qDRFIT_vs_suzuki}}\label{sec:calc_table}
Here, we provide detailed calculations for the numbers presented in Table~\ref{table:qDRFIT_vs_suzuki}. Recall from Eq.~\eqref{eq:qDRIFT_gate_count} that \textsc{qDRIFT} achieves a gate count proportional to
$\lambda^2t^2/\epsilon$, where $\lambda=\sum_j \|h_j \|$ is the 1-norm .
For $p$-th order product formulas, we quote the state-of-the-art Trotter error analysis from Ref.~\cite{childs2019theory} with gate count
\begin{align}
    N =\CO \left( \vertiii{H}_1 \lambda^{1/p} \frac{t^{1+1/p}}{\epsilon^{1/p}} \right).
\end{align}
And for a $k$-local Hamiltonian, the induced 1-norm (using notation from~\cite{childs2019theory})
$\vertiii{H}_1$ is defined as 
\begin{align}
    \displaystyle\vertiii{H}_1: = \max_{\ell}\ \max_{j_\ell} \sum_{j_1,\cdots j_{\ell-1},j_{\ell+1},\cdots, j_k }\norm{ h_{j_1,\cdots,j_\ell,\cdots,j_k}}
\end{align}
which is smaller than the 1-norm $\lambda$.
Here, we also use the convention that  $\CO(\cdot)$ absorbs constants that may depend on  $k, p,\alpha, d$,$Poly \log(n)$  as long as they are independent of system size $n$, time $t$, and the Hamiltonian $H$. For Suzuki formulas values like $p=4$ or $p=6$ are fairly standard. For larger $p$, the constant overheads get prohibitively expensive.
\paragraph*{
\textbf{Nearest Neighbor interaction.}} 
On an integer lattice $[0,R]^d$, the nearest neighbor interacting Hamiltonian reads
\begin{align}
    H = \sum_{\vec{r}\in [0,R]^d} \sum_{ |\vec{r}'-\vec{r}|=1 } h_{\vec{r},\vec{r}'} \quad \text{with} \quad \norm{h_{\vec{r},\vec{r}'}} = \CO(1).
\end{align}
This, in turn, implies
\begin{align}
    \lambda &= \CO(R^d) = \CO(n) \quad \text{and}\\
    \vertiii{H}_1 &= \CO(1).
\end{align}

\paragraph*{
\textbf{Power law interaction.}}
On an integer lattice $[0,R]^d$, the power-law interacting Hamiltonian is 2-local and reads
\begin{align}
    H = \sum_{\vec{r}, \vec{r}'\in [0,R]^d} h_{\vec{r},\vec{r}'} \quad \text{with} \quad  \norm{h_{\vec{r},\vec{r}'}} = \CO(\frac{1}{|\vec{r}'-\vec{r}|^\alpha}).
\end{align}
This, in turn, implies
\begin{align}
   \vertiii{H}_1 &=\begin{cases}
   \CO(1) &\textrm{if}\ \alpha > d \\
   \CO(\log(R))=\CO(\log(n)) &\textrm{if}\ \alpha = d \\
   \CO(R^{d-\alpha})=\CO(n^{1-\alpha/d}) &\textrm{if}\ \alpha < d \\
   \end{cases} \quad \text{and} \\
     \lambda &= \CO(n\vertiii{H}_1).
\end{align}
Note that the $\vertiii{H}_1 $ is essentially the single site energy integral. The total number of terms is $L=\CO(n^2)$.
For $\alpha >2d$, as in~\cite{childs2019theory} one can sacrifice time dependence and truncate terms with $|\vec{r}'-\vec{r}|>\ell$, where $\ell=\CO(nt/\epsilon)^{1/(\alpha-d)}$ is tuned to match the Trotter error with the truncation error.
In this case,
\begin{align}
    L &= \CO(n \ell^{d}) = \CO(n \left(\frac{nt}{\epsilon}\right)^{d/(\alpha-d)}),\\
   \vertiii{H}_1 &= \CO(1)\ \ \ \ (\alpha >2d ), \\
    \lambda &= \CO(n\vertiii{H}_1) =\CO(n),
\end{align}

\paragraph*{\textbf{$k$-local Hamiltonians(SYK-like normalization).}}
Consider an all-to-all interacting $k$-local Hamiltonian~\cite{sachdev1993gapless}
\begin{equation}
    H := \sum_{j_1< \ldots <j_k \le n} h_{j_1\cdots j_k} =\sum_{j_1< \ldots <j_k \le n}  g_{j_1\cdots j_k} Z_{j_1\cdots j_k} \quad \text{with normalization} \quad 
    \norm{Z}= \sqrt{\frac{J^2 (k-1)!}{kn^{k-1}}},
\end{equation}
and $g_{j_1\cdots j_k}$ are random couplings sampled independently from the standard Gaussian distribution,  $Z_{j_1\cdots j_k}$ are determinstic k-local operators, and $J$ is some energy scale.
The number of terms is $L=\binom{n}{k}=\CO(n^k)$ and the normalization is chosen such that 
\begin{align}
    \sum_{j_1< \ldots <j_k \le n} \lV Z_{j_1\cdots j_k}\rV^2 = \CO(n).
\end{align} However, let us drop $g_{j_1\cdots j_k}$ and not bother with the randomness. For the cautious reader, we should call a union bound that $\Pr(\max_{j_1\cdots j_k} \left| g_{j_1\cdots j_k } \right|\ge \epsilon )\le L \e^{-\epsilon^2}$, i.e. we at most lose an extra factor of $\CO(\sqrt{\log(L)})=\CO(\sqrt{\log(n)})$. Therefore, up to a $\sqrt{\log(n)}$ overhead, 
\begin{align}
    \vertiii{H}_1&=  \sum_{j_2,\cdots,j_k }\norm{ h_{j_1,\cdots,j_k}} = \CO( n^{k-1}\cdot n^{-(k-1)/2} ) =\CO(n^{(k-1)/2}),\\
    \lambda &= \CO( n^{k}\cdot n^{-(k-1)/2} ) =\CO(n^{(k+1)/2})
\end{align}
in the calculation we fixed $j_1$ due to symmetry.
\paragraph*{\textbf{$k$-local Hamiltonians (1-norm normalization).}}
In some other settings, one may choose~\cite{Lashkari_2013}
\begin{equation}
    H := \sum_{j_1< \ldots <j_k \le n} h_{j_1\cdots j_k} =\sum_{j_1< \ldots <j_k \le n} Z_{j_1\cdots j_k} \quad \text{with normalization} \quad \norm{Z}= \CO(1/n^{k-1}),
\end{equation}
and without randomly sampled coefficients. The normalization is such that the 1-norm is extensive
\begin{align}
    \sum_{j_1< \ldots <j_k \le n} \lV Z_{j_1\cdots j_k}\rV = \CO(n).
\end{align} 
This in turn implies
\begin{align}
    \vertiii{H}_1 &=  \sum_{j_2,\cdots,j_k }\norm{ h_{j_1,\cdots,j_k}} = \CO( n^{k-1}\cdot n^{-(k-1)} ) =\CO(1),\\
    \lambda  &= \CO( n ).
\end{align}

}

\section{Matrix and vector valued martingales.}\label{sec:intro_martingale}
\AC{
Let $X_1,\ldots,X_N$ be independent and identically distributed (i.i.d.) random variables. 
Then, the strong law of large numbers (LLN) implies that the sample mean $\hat{\mu} =(1/N)\sum_{k=1}^N X_k$ converges to the actual mean $\mu = \mathbb{E}X_k$ ($\hat{\mu} \approx \mu$ as $N \to \infty$).  In fact, even for finite $N$, the sample average is likely to be close to the expectation.  This behavior is called \textit{concentration}. It turns out that concentration is a surprisingly generic phenomenon, and it kicks in earlier than one might expect. 
Asymptotically large sample sizes ($N \to \infty$), which are essential for the law of large numbers and the central limit theorem, are not required establish concentration. An example is Bernstein's inequality for sums of bounded random numbers.
\begin{fact}[Bernstein's inequality]\label{fact:bernstein}
Let $X_1,\ldots,X_N$ be independent random variables that obey $\BE X_i=0$ and $|X_i| \leq R$ almost surely. Then, for $\epsilon >0$,
\begin{align}
    \mathrm{Pr}\left[ \left| \sum^N_{k=1} X_k \right| \ge \epsilon \right] \leq 2\exp \left(\frac{\epsilon^2/2}{v^2+R\epsilon/3}\right) \quad \text{where} \quad v=\sum^N_{k=1} \BE X_k^2.
\end{align}
\end{fact}

Bernstein's inequality equips random sums with a LLN-type concentration behavior already for non-asymptotic sample sizes $N \gtrsim \max \left\{v^2/\epsilon^2,R/\epsilon \right\}$. However, it requires strong assumptions on the random variables involved. They need to be independent, bounded scalars. 
In fact, random processes may concentrate under much weaker conditions.

Martingales form a richer family of processes that capture more realistic random processes and that still enjoy powerful concentration inequalities. Let us offer a minimal technical introduction following Tropp~\cite{tropp2011freedman}.  Consider a filtration of the master sigma algebra $\CF_0\subset \CF_1 \subset \CF_2 \cdots \subset \CF_k \subset \cdots \CF$, where we denote the conditional expectation with respect to $\CF_k$ by the symbol $\BE_k$. A \emph{martingale} is a sequence $\{ B_0, B_1, B_2, \dots \}$ of random variables that satisfies
\begin{align}
    \sigma(B_k) &\subset \CF_k &\textrm{(causality)}, \\
    \BE_{k-1} B_k &:= \BE \left[ B_k | B_{k-1},\ldots,B_0 \right]= B_{k-1} &\textrm{(status quo)}.
\end{align}
Heuristically, we can think of $k$ as a `time' index, and $\CF_k$ contains all possible events that are determined by the history up to time $k$. The present, aka $B_k$ 
may depend on the past (i.e., $B_0,\ldots,B_{k-1}$), but it cannot depend on the future (`causality'). 
The second condition formalizes the requirement that, on average, tomorrow ($B_k$) is the same as today ($B_{k-1}$).

A martingale sequence may be composed of scalars, e.g., a one-dimensional random walk, but we can also consider vector- or matrix-valued martingales. Analyzing concentration for vector- and matrix-valued martingales is an ongoing field of research; for example, see \cite{tropp2011freedman,pinelis2012optimum,HNTR20:Matrix-Product},
but many powerful concentration inequalities already exist.  In this work, we will use the matrix generalization of Freedman's inequality (due to one of the authors). Let $\mathbb{M}_{d \times d}$ denote the space of complex-valued $d \times d$ matrices.

\begin{fact}[{Matrix Freedman~\cite[Corollary~1.3]{tropp2011freedman}}] \label{fact:martingale}
Let $\left\{B_k:k=0,\ldots,\ell, ...\right\} \subset \mathbb{M}_{d \times d}$ be a  matrix martingale.
Assume that the associated difference sequence $C_k = B_k - B_{k-1}$ 
obeys $\|C_k\|\le R $ almost surely.  Define the random variable
\begin{equation}
w_\ell :=\max \left( \left\| \sum_{k=1}^\ell \BE_{k-1} C_k^\dagger C_k 
\right\|, \left\| \sum_{k=1}^\ell \BE_{k-1} C_k C_k^\dagger  \right\| \right).
\end{equation}
Then, for any $\tau >0$
\begin{equation*}
\mathrm{Pr} \left[ \exists \ell : \|B_\ell - B_0 \| \geq \tau, w_\ell \le v \right] 
\leq 2 d \exp \left( \frac{-\tau^2/2}{v +R\tau/3}\right).
\end{equation*} 

\end{fact}

This bound exponentially suppresses the probability of the following undesirable event:  the conditional variance $w$ is small while the actual deviation is large. The intricacy of this bound is that the conditional variance $w_\ell$ is itself a random variable. However, for this work we will use a weaker but more transparent version.

\begin{cor} \label{cor:freedman}
Let $\left\{B_k:k=0,\ldots,N\right\} \subset \mathbb{M}_{d \times d}$ be a matrix martingale.
Assume that the associated difference sequence $C_k = B_k - B_{k-1}$ 
obeys $\|C_k\|\le R $ and its conditional variance obeys $\| \sum_{k=1}^N \BE_{k-1} C_k C_k^\dagger  \| \le v$ almost surely.
Then
\begin{equation*}
\mathrm{Pr} \left[  \|B_N - B_0 \| \geq \tau \right] 
\leq 2 d \exp \left( \frac{-\tau^2/2}{v +R\tau/3}\right)
\quad \text{for any $\tau >0$.}
\end{equation*}

\end{cor}
To arrive at this, ignore the events for $\ell<N$ and use that $w_\ell\le v$ holds almost surely. 
This simplified Matrix Freedman inequality closely resembles Bernstein's inequality. Actually, Fact~\ref{fact:bernstein} can be viewed as a special case of Corollary~\ref{cor:freedman} where $d=1$ and $B_k = \sum_{k'=0}^k X_k$.
But for matrix-valued martingales, the tail bound must depend on the dimension $d$.

Recapitulating the proof of the general Matrix Freedman inequality would go beyond the scope of this work. It makes full use of Lieb's concavity theorem, stopping times, and Burkholder functions.  There is a slightly weaker result that follows from the Golden--Thompson inequality~\cite[Thm.~1.2]{oliveira2010concentration}; see also~\cite[Theorem~11]{Gross2011:Recovering}.  
Similar concentration inequalities are valid for martingales on the complex vector space $\mathbb{C}^d$.   Remarkably, these results are independent of the ambient dimension $d$.
\begin{fact} \label{fact:vector_martingale}
Let $\left\{B_k:k=0,\ldots,N\right\} \subset \mathbb{C}^d$ be a vector martingale taking valules in the inner product space ${\lV \cdot \rV_{\ell_2}}$.  Assume that the associated difference sequence $C_k = B_k - B_{k-1}$ 
obeys $\sum_{k=1}^N \BE_{k-1} \|C_k \|_{\ell_2}^2 \leq v$ and $\|C_k\|_{\ell_2}\le R $ almost surely. Then, for any $\tau >0$
\begin{equation*}
\mathrm{Pr} \left[ \|B_N - B_0 \|_{\ell_2} \geq \tau \right] 
\leq 2 \exp \left( \frac{-\tau^2/2}{v+R\tau/3}\right).
\end{equation*}

\end{fact}
This is simplified version of \cite[Thm. 3.3]{pinelis2012optimum}. To prove the above version, start with~\cite[Thm. 3.1]{pinelis2012optimum} for general martingales on Banach spaces. In our case, vectors are equipped with the standard $\ell_2$ norm, set the smoothness constant $D=1$ (in the notation of~\cite[Thm. 3.1]{pinelis2012optimum}), and optimize for a $\lambda$ like in Bernstein's inequality.

Our concentration analysis of random product formulas with fixed inputs relies on another tool: \textit{uniform smoothness} of the underlying space~\cite{HNTR20:Matrix-Product}.
Given a vector martingale, uniform smoothness supplies a recursive/local bound on moment growth. Decomposing $B_{k}=C_k+B_{k-1}$, all we need is $\BE[C_k|B_{k-1}]=0,$ to apply the following result.

\begin{proposition}[Uniform smoothness] \label{prop:pythagoras}
Let $x,y \in \mathbb{C}^d$ be two random vectors that obey $\mathbb{E} \left[y | x \right] =0$. When $q \geq 2$
\begin{equation*}
\left( \mathbb{E} \|x+y \|_{\ell_2}^q\right)^{2 / q} \leq  \left( \mathbb{E}\|x \|_{\ell_2}^q \right)^{2 / q} + (q-1)
\left( \mathbb{E}\|y \|_{\ell_2}^q \right)^{2 / q}
\end{equation*}
\end{proposition}

A slightly weaker version of this classic fact follows from a short argument.

\begin{proof}
Start with Bonami's inequality~\cite[Cor.~13.1.1]{garling_2007} for normed vector spaces (denoting $\lV \cdot \rV_{\ell_2}$ as $\lV \cdot \rV_2$):
\begin{align}
     \left(\frac{\lV x+y \rV_2^q + \lV x-y \rV_2^q}{2}\right)^{2/q} \le \frac{1}{2} \left(\lV x+\sqrt{q-1}y\rV_2^2+\lV x-\sqrt{q-1}y\rV_2^2\right) = \lV x\rV_2^2+(q-1)\lV y\rV_2^2.
\end{align}
This formula can be converted to the desired statement (Proposition~\ref{prop:pythagoras}) via elementary manipulations. We follow ~\cite[Corollary 4.2]{HNTR20:Matrix-Product}: take expectation and use triangle inequality for $L_{q/2}$ norm
\begin{align}
     \frac{\BE\lV x+y \rV_2^q + \BE\lV x-y \rV_2^q}{2} \le \BE(\lV x\rV_2^2+(q-1)\lV y\rV_2^2)^{q/2}\le \left((\BE\lV x\rV_2^q)^{2/q}+(q-1)(\BE\lV y\rV_2^q)^{2/q}\right)^{q/2}.
\end{align}
Next, following \cite[Proposition~4.3]{HNTR20:Matrix-Product}, by Jensen's in the first and Lyapunov's in the second  inequality,
\begin{align}
     \frac{(\BE\lV x+y \rV_2^q)^{2/q} + (\BE\lV x\rV_2^q)^{2/q}}{2} &\le \frac{(\BE\lV x+y \rV_2^q)^{2/q} + (\BE\lV x-y\rV_2^q)^{2/q}}{2} \\
     &\le
     \left(\frac{\BE\lV x+y \rV_2^q +\BE\lV x-y\rV_2^q}{2} \right)^{2/q} \le
     (\BE\lV x\rV_2^q)^{2/q}+(q-1)(\BE\lV y\rV_2^q)^{2/q}.
\end{align}
By rearranging terms, we obtain the result with the constant $2(q-1)$, which is off by a factor of 2. The advertised constant (Proposition~\ref{prop:pythagoras}) can be obtained via a more involved trick \cite[Lemma~A.1]{HNTR20:Matrix-Product}.
Geometrically, this result expresses the uniform smoothness of the space $(\BE\lV \cdot \rV_{\ell_2}^q)^{1/q}$.
\end{proof}

We refer to~\cite{HNTR20:Matrix-Product} for further exposition on uniform smoothness for general Schatten $p$-norm. Recently, this method has been applied to dynamic properties of random Hamiltonians~\cite{chen2021concentration}.
For the task at hand, we can use Markov's inequality to convert such bounds on moment growth into a strong tail bound, similar to Fact~\ref{fact:vector_martingale}. This is the content of Lemma~\ref{lem:fixed-input-fluctuation} below. 
}

\section{Technical details and proofs}
\label{sec:proofs}

\subsection{
Proof of Theorem~\ref{thm:allinput}~and~Corollary~\ref{cor:exp_allinput}: Approximation error under the worst-possible input}
\label{sec:proofthm1}

The proofs of Theorem~\ref{thm:allinput} and Corollary~\ref{cor:exp_allinput} were sketched in Section~\ref{sec:proofidea}. 
This section contains the details. In Section~\ref{sub:conversion}, we first relate the diamond distance to the operator norm. This allows us to work with the operator norm, which is mathematically simpler. Then we bound the two error contributions arising from the deterministic bias (in Section~\ref{sub:bias}), as well as random fluctuations (in Section~\ref{sub:fluctuations}). Finally, we combine the two
bounds to obtain a convergence guarantee for randomly sampled product formulas. This is the content of Section~\ref{sub:expectation}.

\subsubsection{Conversion from diamond distance into operator norm}
\label{sub:conversion}

The diamond distance is a rather intricate object. Although it can be phrased implicitly as a semidefinite program, analytical formulas are rare and far between. It is, however, sometimes possible to relate diamond distances to other figures of merit that are easier to access.\AC{
The mixing lemma by Campbell~\cite{campbell_mixing16} and Hastings~\cite{hastings2016turning} is one very useful example.

\begin{lemma}[Mixing lemma~\cite{campbell_mixing16,hastings2016turning}]\label{lem:mixing}
Let $U$ be a fixed unitary and $V$ be a random approximation thereof that obeys 
$\lV V-U \rV \le a$ almost surely. Then, the associated channels $\mathcal{U}(X)=UXU^\dagger$ and $\mathcal{V}(X) = V X V^\dagger$ obey
\begin{equation*}
    \tfrac{1}{2} \lV \BE \mathcal{V} -\mathcal{U} \rV_{\diamond} \le  \lV \BE V - U \rV+\tfrac{1}{2}a^2.
\end{equation*}
\end{lemma}
}
More recent insights, in particular~\cite[Thm.~3.56]{watrous2018book}, allow for sharpening this bound. In particular, we can completely remove the $a^2/2$-term on the r.h.s.

\begin{lemma} \label{lem:diamond-to-operator}
Let $\mathcal{U}(\rho) = U \rho U^\dagger$ and $\mathcal{V}(\rho) =V \rho V^\dagger$ be unitary channels. Then,
$\tfrac{1}{2}\| \mathcal{U}(|\psi \rangle \! \langle \psi|) - \mathcal{V} (|\psi \rangle \! \langle \psi|) \|_1 \leq \| (U-V) | \psi \rangle \|_{\ell_2}$ for any pure state $|\psi \rangle$.
In turn, $\tfrac{1}{2} \| \mathcal{U} - \mathcal{V} \|_\diamond \leq \|U - V \| $.
The latter relation can be generalized to averages of random unitary channels with an additional factor of $2$:
\begin{equation*}
\tfrac{1}{2}\| \mathcal{U} - \mathbb{E}[\mathcal{V}] \|_\diamond \leq 2 \|U - \mathbb{E}[V] \|.
\end{equation*}
\end{lemma}

The first claim follows from the fact that stabilization is not required to compute the diamond distance between two unitary channels~\cite[Sec.~5.3]{AKN98:Quantum-Circuits}. The second claim is more complicated. The original claim without the factor of $2$ on the right hand side turned out to be incorrect. Richard Allen and Haimeng Zhao found an elegant proof showing that the second claim remains correct with the additional factor of $2$.

\begin{proof}
Fix an input state $|\psi \rangle$ and denote the output state vectors by $|u \rangle=U| \psi \rangle$ and $|v \rangle = V | \psi \rangle$. Normalization ensures that $|\langle u,v \rangle | \leq 1$. Applying the Fuchs--van de Graaf relations~\cite[Theorem~3.33]{watrous2018book} to convert the output trace distance into a (pure) output fidelity gives:
\begin{align*}
\tfrac{1}{2} \| |u \rangle \! \langle u| - |v \rangle \! \langle v| \|_1 = \sqrt{1-| \langle u,v \rangle|^2} = \sqrt{(1+| \langle u,v \rangle|) (1-| \langle u,v\rangle|)} \leq \sqrt{2(1-\mathrm{Re}(\langle u,v \rangle))} = \| |u \rangle - |v \rangle \|_{\ell_2}. 
\end{align*}
Since stabilization is not necessary for computing the diamond distance of two unitary channels~\cite[Sec.~5.3]{AKN98:Quantum-Circuits}, this relation immediately implies the first diamond distance bound:
\begin{equation*}
\tfrac{1}{2} \| \mathcal{U} - \mathcal{V} \|_\diamond = \max_{|\psi \rangle} \tfrac{1}{2} \| \mathcal{U}(|\psi \rangle \! \langle \psi|) - \mathcal{V}(|\psi \rangle \! \langle \psi| ) \|_1 \leq \max_{| \psi \rangle}\|(U-V)| \psi \rangle \|_{\ell_2} = \|U-V \| .
\end{equation*}
This completes the first claim.

For the second claim with $V$ being random, we let $E = \mathbb{E}[V]$. For any pure state $|\psi \rangle$, we expand the action of the expected channel by centering it around $E$:
\begin{equation*}
\mathbb{E}[V|\psi \rangle\!\langle \psi|V^\dagger] = \mathbb{E}[((V-E)+E)|\psi \rangle\!\langle \psi|((V-E)+E)^\dagger] = E|\psi \rangle \! \langle \psi|E^\dagger + \mathbb{E}[(V-E)|\psi \rangle\!\langle \psi|(V-E)^\dagger] ,
\end{equation*}
where the cross terms vanish because $\mathbb{E}[V-E] = 0$. By the triangle inequality, we have:
\begin{equation*}
\tfrac{1}{2}\| U|\psi \rangle\!\langle \psi|U^\dagger - \mathbb{E}[V|\psi \rangle\!\langle \psi|V^\dagger] \|_1 \leq \tfrac{1}{2}\| U|\psi \rangle\!\langle \psi|U^\dagger - E|\psi \rangle\!\langle \psi|E^\dagger \|_1 + \tfrac{1}{2}\| \mathbb{E}[(V-E)|\psi \rangle\!\langle \psi|(V-E)^\dagger] \|_1 .
\end{equation*}
For the first term, we rewrite the argument as $(U-E)|\psi \rangle\!\langle \psi|U^\dagger + E|\psi \rangle\!\langle \psi|(U^\dagger-E^\dagger)$. Applying the triangle inequality and the trace norm property $\| |\phi \rangle\!\langle \varphi| \|_1 = \| |\phi \rangle \|_2 \| |\varphi \rangle \|_2$, we obtain:
\begin{equation*}
\tfrac{1}{2}\| U|\psi \rangle\!\langle \psi|U^\dagger - E|\psi \rangle\!\langle \psi|E^\dagger \|_1 \leq \tfrac{1}{2} \left( \|(U-E)|\psi \rangle \|_2 \|U|\psi \rangle \|_2 + \|E|\psi \rangle \|_2 \|(U-E)|\psi \rangle \|_2 \right) .
\end{equation*}
Since $U$ is unitary, $\|U|\psi \rangle \|_2 = 1$. By the convexity of the vector norm, $\|E|\psi \rangle \|_2 = \|\mathbb{E}[V]|\psi \rangle \|_2 \leq \mathbb{E}[\|V|\psi \rangle \|_2] = 1$. Thus, the bias term is bounded by $\|(U-E)|\psi \rangle \|_2$.
For the second term, observe that $(V-E)|\psi \rangle\!\langle \psi|(V-E)^\dagger$ is positive semi-definite, and so is its expectation. Thus, its trace norm equals its trace. Using $V^\dagger V = I$, we compute:
\begin{equation*}
\| \mathbb{E}[(V-E)|\psi \rangle\!\langle \psi|(V-E)^\dagger] \|_1 = \mathrm{Tr}\big(\mathbb{E}[(V-E)|\psi \rangle\!\langle \psi|(V-E)^\dagger]\big) = \langle \psi| (I - E^\dagger E) |\psi \rangle .
\end{equation*}
Using $U^\dagger U = I$, we rewrite $I - E^\dagger E = U^\dagger(U-E) + (U^\dagger - E^\dagger)E$. Applying the Cauchy-Schwarz inequality yields:
\begin{equation*}
\langle \psi| U^\dagger(U-E) + (U^\dagger - E^\dagger)E |\psi \rangle \leq \|U|\psi \rangle \|_2 \|(U-E)|\psi \rangle \|_2 + \|E|\psi \rangle \|_2 \|(U-E)|\psi \rangle \|_2 \leq 2 \|(U-E)|\psi \rangle \|_2 .
\end{equation*}
Multiplying by $\tfrac{1}{2}$, the variance term is also bounded by $\|(U-E)|\psi \rangle \|_2$.
Combining the two terms, we obtain:
\begin{equation*}
\tfrac{1}{2}\| U|\psi \rangle\!\langle \psi|U^\dagger - \mathbb{E}[V|\psi \rangle\!\langle \psi|V^\dagger] \|_1 \leq 2 \|(U-E)|\psi \rangle \|_2 .
\end{equation*}
Finally, we extend this to the diamond norm. By definition, the diamond distance is maximized by some pure state $|\Psi \rangle$ on the extended space $\mathcal{H} \otimes \mathcal{H}_{\mathrm{anc}}$:
\begin{equation*}
\tfrac{1}{2}\| \mathcal{U} - \mathbb{E}[\mathcal{V}] \|_\diamond = \max_{|\Psi \rangle} \tfrac{1}{2}\| (U \otimes I)|\Psi \rangle\!\langle \Psi|(U \otimes I)^\dagger - \mathbb{E}[(V \otimes I)|\Psi \rangle\!\langle \Psi|(V \otimes I)^\dagger] \|_1 .
\end{equation*}
Applying our pure state bound to the extended unitaries $U \otimes I$ and $V \otimes I$, we find:
\begin{equation*}
\tfrac{1}{2}\| \mathcal{U} - \mathbb{E}[\mathcal{V}] \|_\diamond \leq \max_{|\Psi \rangle} 2 \| ((U-E) \otimes I) |\Psi \rangle \|_2 = 2 \| (U-E) \otimes I \| = 2 \| U - E \| ,
\end{equation*}
which completes the proof.
\end{proof}

\subsubsection{Controlling the deterministic bias} \label{sub:bias}

Next, we establish a bound on the deterministic bias between the averaged channel
and the ideal unitary evolution.

\begin{proposition} \label{prop:bias1}
Consider the i.i.d.~unitary product constructed by the \textsc{qDRIFT} protocol \eqref{eq:qDRIFT} for
simulating $U=\exp (-\mathrm{i} t H)$ for evolution time $t$.  Define the total strength $\lambda= \sum_j\norm{h_j} $. Then
\begin{equation*}
\| U - \mathbb{E} [V_N \cdots V_1] \|  \leq \frac{t^2 \lambda^2}{N}.
\end{equation*}
\end{proposition}

Note that the improved mixing lemma, Lemma~\ref{lem:diamond-to-operator} above, allows for converting this statement into a diamond distance bound for the associated channels:
\begin{equation}
\tfrac{1}{2}\| \mathcal{U}-\mathbb{E} [\mathcal{V}_N \circ \cdots \circ \mathcal{V}_1]\|_\diamond 
\leq 2  \|U - \mathbb{E} \left[V_N \cdots V_1 \right] \| 
\leq
\frac{2 t^2 \lambda^2}{N}. \label{eq:campbell}
\end{equation}
The proof of Proposition~\ref{prop:bias1} is based on an extension of the numerical bounds $|\mathrm{e}^{\mathrm{i}x}-1| \leq |x|$ and
$|\mathrm{e}^{\mathrm{i}x}-\mathrm{i}x - 1 | \leq x^2/2$ for all $x \in \mathbb{R}$ to Hermitian matrices.

\begin{fact} \label{fact:taylor}
Let $X$ be Hermitian.  Then we have the zero-order bound
$\|\exp (\mathrm{i}X)-\mathbb{I} \| \leq \|X\|$ and 
the first-order bound
$
\| \exp (\mathrm{i}X)-\mathrm{i}X-\mathbb{I} \| \leq \tfrac{1}{2} \|X\|^2$.
\end{fact}

These observations can be converted into accurate operator-norm bounds for the expected error of individual \textsc{qDRIFT} steps.

\begin{lemma} \label{lem:taylor}
Fix a Hamiltonian $H=\sum_{j=1}^L h_j$ and parameters $N, t$.
Set $U^{1/N}=\exp \left( - \mathrm{i} (t/N) H \right)$ and $\lambda = \sum_{j=1}^L \|h_j\|$. 
Then, the random matrix $V$ defined in~\eqref{eq:qDRIFT} obeys
\begin{equation*}
\text{ (almost surely)}
\quad \text{and}
\quad 
\| (\mathbb{E} V) - U^{1/N}\| \leq \frac{t^2 \lambda^2}{N^2}
\end{equation*}
\end{lemma}
\begin{proof}
Streamline the notation from \textsc{qDRIFT}(Algorithm~\ref{alg:qdrift}) by absorbing the scaling factor $(t/N)$ into the random Hermitian matrix $X$.  In particular,
$V=\exp(-\mathrm{i}X)$, $\mathbb{E}V=\mathbb{E}[\exp(-\mathrm{i}X)]$, $U^{1/N}=\exp(-\mathrm{i} \mathbb{E}[X])$ and $\norm{X}=(t\lambda)/N$ almost surely.
Observe that
\begin{align*}
\|V-(\mathbb{E}V)\| \leq \|\exp(-\mathrm{i}X)-\mathbb{I} \|+ \| \mathbb{I}-\mathbb{E}[\exp (-\mathrm{i}X)]\|
\leq \| \exp(-\mathrm{i}X) - \mathbb{I} \| + \mathbb{E} \|\mathbb{I}-\exp(-\mathrm{i}X) \|,
\end{align*}
where the last inequality is Jensen's.
Fact~\ref{fact:taylor} and uniform normalization then imply $\|\exp(-\mathrm{i}X)-\mathbb{I} \| \leq \|X\|= (t\lambda)/N$ for any instance of the random matrix $X$. 
This uniform bound also covers the expected norm difference and we conclude $\|V-(\mathbb{E}V) \| \leq 2 t \lambda/N$. 
The (tighter) second claim can be derived in a similar fashion.
A combination of Jensen's inequality, Fact~\ref{fact:taylor}, and uniform normalization delivers
\begin{align*}
\| (\mathbb{E}V)-U^{1/N}\|
&= \| \mathbb{E}\left[ \exp (-\mathrm{i}X)]-\mathbb{I} +\mathrm{i}X \right] + \left( \mathbb{I}-\mathrm{i} \mathbb{E}[X]-\exp(-\mathrm{i} \mathbb{E}[X]) \right) \| \\
&\leq \mathbb{E} \| \exp (-\mathrm{i}X) -\mathbb{I} +\mathrm{i}X \|
+ \| \exp (-\mathrm{i} \mathbb{E}[X])-\mathbb{I} + \mathrm{i} \mathbb{E}[X]\| \\
&\leq \tfrac{1}{2} \mathbb{E} \|X \|^2 + \tfrac{1}{2} \| \mathbb{E}[X]\|^2
\leq \mathbb{E} \|X \|^2
= \left( t \lambda/ N \right)^2.
\end{align*}
This is the advertised result.
\end{proof}

We also need a statement regarding error accumulation over several applications of similar, but not identical, linear operators.
It is a rather intuitive consequence of operator norm sub-multiplicativity and the triangle inequality.
See~\cite{suzuki1985decomposition} for related results.

\begin{lemma} \label{lem:composition}
Let $\mathbb{E}V$ and $U^{1/N}$ be matrices with bounded operator norm, i.e.\ $\|\mathbb{E}V\| \leq 1$ and $\|U^{1/N}\| \leq 1$. Then
\begin{equation*}
\|(\mathbb{E}V)^N-(U^{1/N})^N \| \leq N \|(\mathbb{E}V)-U^{1/N} \|.
\end{equation*} 
\end{lemma}

\begin{proof}
The triangle inequality and sub-multiplicativity imply
\begin{align*}
\|A_1 A_2 - B_1 B_2 \|
&= \| (A_1 - B_1) A_2 + B_1 (A_2-B_2) \|
\leq \|A_2 \| \|A_1 - B_1 \| + \|B_1 \| \|A_2 - B_2 \|
\end{align*}
for any matrix quadruple $A_1,A_2,B_1,B_2$ with compatible dimensions. 
Use the assumed operator norm bounds to iteratively apply this relation and deduce the statement:
\begin{align*}
\| (\mathbb{E}V)^{N}-(U^{1/N})^N \|
&= \| (\mathbb{E} V) (\mathbb{E}V)^{N-1}- U^{1/N}(U^{1/N})^{N-1}\| \\
&\leq \| (\mathbb{E}V)-U^{1/N}\| + \| (\mathbb{E}V)^{N-1}- (U^{1/N})^{N-1}\|
\leq \cdots \leq  N \| (\mathbb{E}V)-U^{1/N}\|.
\end{align*}
This is the stated result.
\end{proof}

\begin{proof}[Proof of Proposition~\ref{prop:bias1}]
The main result of this section immediately follows from combining Lemma~\ref{lem:composition} and Lemma~\ref{lem:taylor}.
Decompose $U=\exp (-\iunit tH)$ into $N$ steps $U^{1/N}=\exp (-\mathrm{i} (t / N) H)$ and conclude
\begin{align*}
\| U - (\mathbb{E} V)^N \|
= \|(U^{1/N})^N - (\mathbb{E}V)^N \| \leq N \| U^{1/N}-(\mathbb{E} V)\|
\leq \frac{t^2 \lambda^2}{N}.
\end{align*}
This is what we had to show.
\end{proof}

\subsubsection{Controlling random fluctuations} \label{sub:fluctuations}

In the previous subsection we have essentially recapitulated the state of the art regarding \textsc{qDRIFT}: the algorithm provides an accurate approximation in expectation over all possible random choices (deterministic bias).
In this section, things start to get interesting. 
We want to show that a single realization of \textsc{qDRIFT} is likely to provide a good approximation, provided that the number of steps $N$ is sufficiently large. In order to achieve this goal, we need to show that concrete fluctuations around the (accurate) expected behavior remain small: 
\begin{equation}
V_N \cdots V_1 \approx \mathbb{E}[V_N \cdots V_1] =(\mathbb{E}V)^N\quad \text{with high probability.}
\label{eq:fluctuation}
\end{equation}
In words, we need to show that a product of i.i.d.~random matrices concentrates sharply around its expectation value. 
This is an interesting and nontrivial problem, even in the (asymptotic) large $N$ limit. While sharp concentration bounds for sums of i.i.d.~random matrices have been available for more than a decade now \cite{AW02:Strong-Converse,Tro12:User-Friendly}, our understanding of concentration for random matrix products is more limited; see~\cite{HNTR20:Matrix-Product} and references therein. There is a lot of math literature on random walks on Lie groups, but the focus is usually on asymptotic convergence and the machinery is different; see~\cite{varj2012random} and references therein.  The small-step regime has seen less development, although there are some asymptotic bounds~\cite{versendaal2019large}.

Fortunately, the \textsc{qDRIFT} construction has several appealing features: the random unitaries $V_N,\ldots,V_1$ are i.i.d.~unit-norm matrices that are close to the identity matrix ($\|V - \mathbb{I}\| \leq t\lambda/N$ almost surely) and close to their expectation ($\|V-(\mathbb{E}V)\|\leq 2 t \lambda/N$ almost surely).
These properties allow us to use the matrix martingale formalism to derive a strong, nonasymptotic result on the quality of the approximation.

\begin{proposition}[\textsc{qDRIFT}: Operator norm concentration] \label{prop:fluctuation}
Consider a Hamiltonian $H=\sum_{j=1}^L h_j$ with interaction strength $\lambda = \sum_{j=1}^L \|h_j\|$, and fix parameters $N, t$. Suppose that $V_N,\ldots,V_1$ are i.i.d.~instances of the random unitary $d \times d$ matrix $V$ constructed by the \textsc{qDRIFT} protocol~\eqref{eq:qDRIFT}. Then 

\begin{equation*}
\mathrm{Pr}
\left[ \| V_N \cdots V_1 - \mathbb{E}[V_N \cdots V_1]\| \geq 
{\epsilon/2} \right]
\leq 2 d \exp \left( - \frac{N \epsilon^2}{44 t^2 \lambda^2}\right) \quad \text{for any $\epsilon \in [0,4t \lambda]$}.
\end{equation*}
In particular, $N \geq (44t^2 \lambda^2/\epsilon^2) \log (2d / \delta)$ implies that $\| V_N \cdots V_1 - \mathbb{E}[V_N \cdots V_1] \| \leq \epsilon/2$ with probability at least $1-\delta$.
\end{proposition}

This statement provides a strong tail bound for random fluctuations in the small-error regime
$\epsilon \leq 4t \lambda$. 
As $N$ increases, the probability of incurring (at least) error $\epsilon/2$ diminishes exponentially.
For $\epsilon > 4t \lambda$, we have instead a subexponential tail bound:
$\mathrm{Pr} \left[ \| V_n \cdots V_1 - \mathbb{E}[V_N \cdots V_1] \| \geq \tau \right]\leq 2d \exp (- N \epsilon/ 6t \lambda )$.
We refer to~\eqref{eq:fluctuation-bound-detail} for a unified statement that covers both regimes.

The proof technique deserves some exposition, as it is rather general and may be of independent interest.
It heavily utilizes the martingale concentration tools introduced in Appendix~\ref{sec:intro_martingale}.
For fixed $N$, 
we interpolate between both sides of Rel.~\eqref{eq:fluctuation} by means of a random process $\left\{B_k: k=0,\ldots,N\right\}$:
\begin{equation*}
B_k = (\mathbb{E} V )^{N-k}V_k \cdots V_1 \quad \text{such that}\quad B_0 = (\mathbb{E}V)^N \text{ and } B_N = V_N \cdots V_1.
\end{equation*}
The increments of this random process are certainly not independent. For instance, $B_{k}$ depends on the (random) choice of $V_{k}$ and \textit{all} previous choices $V_{k-1}, \ldots, V_1$. 
Fortunately, we can recognize it as a (matrix-valued) martingale satisfying the defining properties:
\begin{enumerate}
\item \textit{Causality:} Each $B_k$ is completely determined by the information we have collected up to step $k$. That is, the (random) choices of $V_k,\ldots,V_1$.
\item \textit{Status quo:} Conditioned on previous choices, the expectation of $B_{k+1}$ equals $B_k$: for $1 \leq k \leq N$
\begin{equation}
\mathbb{E} \left[B_{k+1}|V_k \cdots V_1 \right]
= (\mathbb{E} V)^{N-(k+1)}\mathbb{E}_{V_{k+1}}\left[V_{k+1}\right] V_k \cdots V_1 = (\mathbb{E} V)^{N-k}V_k \cdots V_1 = B_k.
\label{eq:martingale}
\end{equation}
This feature underscores similarities to an unbiased random walk.  On average, ``tomorrow'' ($B_{k+1}$) is the same as ``today'' ($B_k$).
\end{enumerate}

With this matrix martingale reformulation at hand, we can prove Proposition~\ref{prop:fluctuation} using a concentration inequality for matrix martingales, namely Corollary~\ref{cor:freedman}.

\begin{proof}[Proof of Proposition~\ref{prop:fluctuation}]

We have already established that the random process $B_k = (\mathbb{E}  V )^{N-k} V_k \cdots V_1$ forms a matrix martingale that interpolates between $B_0 = \mathbb{E}[V_N \cdots V_1] = (\mathbb{E} V)^N$ and $V_N = V_N \cdots V_1$.
The elements of the associated difference sequence are
\begin{align*}
C_k &= B_k - B_{k-1}= (\mathbb{E}V)^{N-k}\left( V_k - (\mathbb{E} V_k)\right) V_{k-1} \cdots V_1
\quad \text{with} \quad k=1,\ldots,N.
\end{align*}
Recall that $V_k = \exp(-\mathrm{i} X_j)$ for some Hermitian matrices $X_j = \tfrac{t}{N}\tfrac{\lambda}{\|h_j\|}h_j$ with index $1 \leq j \leq L$. Boundedness ($\| \mathbb{E} V \| , \| V_k \| \leq 1)$, Fact~\ref{fact:taylor} ($\| \exp(-\mathrm{i} X)-\mathbb{I} \| \leq \|X \|$ for $X$ Hermitian), \AC{and Lemma~\ref{lem:taylor} ensure
\begin{align*}
\|C_k \| &\leq  \| \mathbb{E} V\|^{N-k}\| V_k - (\mathbb{E}V_k) \| \|V_{k-1} \cdots V_1 \|
\leq \| V_k - (\mathbb{E}V_k)\| = \frac{2t \lambda}{N}
\end{align*}}
almost surely. 
Set $R=2 t \lambda/N$, and 
invoke Corollary~\ref{cor:freedman}
to conclude that
\begin{align}
\mathrm{Pr} \left[ \|B_N - B_0 \| \geq \tau \right]
\leq 2d \exp \left( - 
\frac{N\tau^2}{8 (t \lambda)^2+4(t \lambda) \tau/3}\right).
\label{eq:fluctuation-bound-detail}
\end{align}
The statement follows from bounding the somewhat complicated exponential by either $\exp \left(-{3 \tau^2}/(8NR^2)\right)$ for $\tau \leq 2 \lambda t$ or by $\exp \left( - 3 \tau^2/(8R) \right)$ for $\tau \geq 2 \lambda t$.  Last, we substitute $\tau = \epsilon/2$.
\end{proof}

In fact, the same proof works for any adapted small-step random walks on the unitary group. Such a generalization results in Theorem~\ref{thm:summary_of_errors} and we refer to Appendix~\ref{sec:proofgeneral} for details.

\subsubsection{A bound for expected errors} \label{sub:expectation}

In the previous subsection, we established that a sufficiently long \textsc{qDRIFT} random product formula
concentrates sharply around its expectation.  
We can translate this statement into a bound on the expected fluctuation
around the true evolution.

\begin{proposition}[\textsc{qDRIFT}: Expected diamond norm error] \label{prop:expectation}
Consider an $n$-qubit Hamiltonian $H=\sum_{j=1}^L h_j$ with total strength $\lambda = \sum_{j=1}^L \|h_j\|$.
Fix parameters $N, t$, and assume that $N \geq n$.
Set $\mathcal{U}=U X U^\dagger$ with $U = \exp \left( -\mathrm{i} t H \right)$, and suppose that $\mathcal{V}_N,\ldots,\mathcal{V}_1 \sim\mathcal{V}$
are i.i.d.~realizations of the \textsc{qDRIFT} protocol.  That is, $\mathcal{V}(X) = V X V^\dagger$,
where $V$ is defined by~\eqref{eq:qDRIFT}. Then
\begin{equation}
\mathbb{E} \big[ \tfrac{1}{2}\| \mathcal{U}-\mathcal{V}_N \circ \cdots \circ \mathcal{V}_1 \|_\diamond \big] \leq 
\frac{t^2 \lambda^2}{N}+ C \frac{nt \lambda}{N}+C \sqrt{\frac{n t^2 \lambda^2}{N}} \approx C \sqrt{\frac{n t^2 \lambda^2}{N}},
\label{eq:expected-error}
\end{equation}
where $C >0$ is a (modest) numerical constant.
The symbol $\approx$ denotes an accurate approximation in the large-$N$ regime.
\end{proposition}

It is instructive to compare this assertion to the original \textsc{qDRIFT} result \cite{campbell2019random}
and Eq.~\eqref{eq:campbell}:
\begin{equation*}
\tfrac{1}{2}\| \mathcal{U} - \mathbb{E} \left[ \mathcal{V}_N \circ \cdots \circ \mathcal{V}_1 \right]\|_\diamond \leq \frac{2 t^2 \lambda^2}{N}.
\end{equation*}
Note that the expectation over all possible realizations of all $N$ unitary channels appears inside the diamond distance. This implies that \textsc{qDRIFT} performs well on average over many random realizations, provided that the number $N$ of steps exceeds $t^2 \lambda^2/\epsilon$. In contrast,~\eqref{eq:expected-error} has the expectation outside the diamond distance.

Our result gives a much stronger conclusion: An \textit{individual} realization of the randomized \textsc{qDRIFT} protocol does not deviate much from the target evolution, for any input states and observables, with very high probability.  The price for such an improvement is a larger number of steps that depends on the system size.  For $n$ qubits, the gate complexity $N \geq C n t^2\lambda^2/\epsilon^2$ is sufficient to ensure $\epsilon$-closeness on average. The quadratic scaling in the accuracy parameter $\epsilon$ is necessary (for large $N$) because of the central limit theorem for martingales.

The appearance of the number $n$ of qubits is a consequence of measuring closeness in diamond distance.
To obtain
\begin{equation*}
\epsilon \geq
\mathbb{E}\big[ \tfrac{1}{2}\| \mathcal{U}-\mathcal{V}_N \circ \cdots \circ \mathcal{V}_1 \|_\diamond \big]
    = \mathbb{E} \big[ \tfrac{1}{2} \max_{\rho \text{ state}} \| \mathcal{U}_\rho - \mathcal{V}_N \circ \cdots \circ \mathcal{V}_1 (\rho) \|_1 \big],
\end{equation*}
we need the random product formula to behave for all possible $n$-qubit input states $\rho$ simultaneously.
If we restrict our attention to any fixed input state, we can obtain a gate complexity that does not
depend on $n$.  This is the topic of the next section.

\begin{proof}[Proof of Proposition~\ref{prop:expectation}]
First, we relate the expected diamond distance to an expected operator norm distance and split it up into deterministic bias and (expected) fluctuations:
\begin{equation*}
\mathbb{E}\big[ \tfrac{1}{2} \| \mathcal{U} - \mathcal{V}_N \circ \cdots \circ \mathcal{V}_1 \|_\diamond \big] \leq \norm{U - (\mathbb{E}V)^N } + \mathbb{E} \norm{V_N \cdots V_1 - (\mathbb{E}V)^N}
\end{equation*}
The first term is deterministic and controlled by Proposition~\ref{prop:bias1}: $\| U -(\mathbb{E} V)^N \| \leq t^2 \lambda^2/N$.
The second term can be bounded by integrating the tail bound in Proposition~\ref{prop:fluctuation},  or rather the tighter bound presented~\eqref{eq:fluctuation-bound-detail}; see~\cite[Remark~6.5]{Tro12:User-Friendly}.
For $n$ qubits, we have $d=2^n$ and \AC{conclude
\begin{align*}
\mathbb{E} \| V_N \cdots V_1 - (\mathbb{E} V)^N \| &= \int_{0}^\infty \mathrm{Pr} \left[ \| V_N \cdots V_1 - (\mathbb{E} V)^N \| \geq \tau \right] \mathrm{d}\tau \\
&\leq \int_0^\infty \min\left(1,2 \cdot 2^n \exp \left( - \frac{N\tau^2/2}{4 \lambda^2 t^2 +2 \lambda t \tau/3} \right)\right) \mathrm{d} \tau \\
&\leq C\max \left\{\sqrt{ \frac{nt^2 \lambda^2}{N}}, \frac{nt \lambda}{N} \right\} 
\leq 2C\left( \frac{n t \lambda}{N}+ \sqrt{\frac{nt^2 \lambda^2}{N}} \right).
\end{align*}
Here, we have evaluated the integral by two segments: the integrand is at most $1$ until $\tau\sim \max( \sqrt{nt^2\lambda^2/N}, n t \lambda/N )$; for larger $\tau$, the expression decays exponentially and 
integrating it only produces a contribution of order $\mathcal{O}(\max( \sqrt{t^2\lambda^2/N}, t \lambda/N ))$. Also, $2C$ absorbs all constants.}
\end{proof}

\subsection{
Proof of Theorem~\ref{thm:fixedinput}: Approximation error under a single  input}
\label{sec:proofthm2}

Proposition~\ref{prop:expectation} asserts that a single, random realization of the \textsc{qDRIFT} protocol \eqref{eq:qDRIFT} accurately approximates a unitary target evolution with respect to the diamond norm: 
\begin{equation*}
\mathbb{E} \big[ \tfrac{1}{2} \| \mathcal{U} - \mathcal{V}_N \circ \cdots \circ \mathcal{V}_1 \|_\diamond \big]
    = \mathbb{E} \big[ \tfrac{1}{2} \max_{\rho \text{ state}} \| \mathcal{U}(\rho) - \mathcal{V}_N \circ \cdots \circ \mathcal{V}_N (\rho) \|_1 \big] \lesssim C \sqrt{\frac{n t^2 \lambda^2}{N}}.
\end{equation*}
Here, $\lesssim$ denotes an accurate approximation of the true bound in the large $N$ regime. This bound scales linearly in the (qubit) system size $n$.
The dependence on $n$ should not come as a surprise, since the diamond norm produces a very stringent worst-case distance measure. As emphasized by the above reformulation, the approximation must be accurate even when we optimize to find the worst possible input state $\rho$. 

In Hamiltonian simulation, demanding such a stringent worst-case promise may be excessive. In most practical applications, the input state $\rho$ is fixed and simple, e.g., a product state. 
In this more practical setting, we can obtain a gate complexity $N$ that does not depend on the system size $n$.
The main result of this section asserts
\begin{equation*}
\max_{\rho \text{ state}} \mathbb{E} \big[ \tfrac{1}{2} \| \mathcal{U}(\rho) - \mathcal{V}_N \circ \cdots \circ \mathcal{V}_1 (\rho) \|_1 \big] \leq C \sqrt{ \frac{t^2 \lambda^2}{N} }.
\end{equation*}
In other words, fixing an arbitrary input state $\rho$  helps a lot. A total number of $N = 4 \left( t \lambda/\epsilon\right)^2$ steps ensures that \textsc{qDRIFT} produces an $\epsilon$-accurate output state, with respect to trace distance.

The proof is similar in spirit to the argument behind Proposition~\ref{prop:expectation}.   We construct a vector martingale that describes the evolution of the state.  We control the behavior of this martingale using the uniform smoothness of the $L_q(\ell_2)$ norm.
This argument is inspired by the work \cite{HNTR20:Matrix-Product} on concentration of random matrix products.

\subsubsection{Approximation error for a fixed state}

In this section, we state and prove our main technical result
on the action of the \textsc{qDRIFT} protocal on a fixed input state.

\begin{proposition}[\textsc{qDRIFT}: Action on a fixed state] \label{prop:fixed-input}
Consider a Hamiltonian $H=\sum_{j=1}^L h_j$ with total strength $\lambda = \sum_{j=1}^L \|h_j\|$.
Fix evolution time $t$ and a gate count $N$ that obeys $N \geq (t \lambda)^2$. Let $\mathcal{V}_1, \dots, \mathcal{V}_N$ be the i.i.d.~random unitary evolution
operators constructed by the \textsc{qDRIFT} protocol~\eqref{eq:qDRIFT}.
Then,
\begin{equation*}
\max_{\rho\ \mathrm{state}}
    \mathbb{E}\big[ \tfrac{1}{2} \big\| \mathcal{U}(\rho) - \mathcal{V}_N \circ \cdots \mathcal{V}_1(\rho) \big\|_1 \big]
    \leq 4 \sqrt{\frac{t^2\lambda^2}{N}}.
\end{equation*}
Moreover, for $\epsilon > 0$,
\begin{equation*}
\max_{\rho\ \mathrm{state}}
    \mathrm{Pr}\big[ \tfrac{1}{2} \big\| \mathcal{U}(\rho) - \mathcal{V}_N \circ \cdots \mathcal{V}_1(\rho) \big\|_1
    > \epsilon \big]
    \leq \exp\left( \frac{-\epsilon^2 N}{32\mathrm{e} t^2 \lambda^2} \right).
\end{equation*}
\end{proposition}

\begin{proof}
First, we reduce the problem to a question about pure states.
For any $q \geq 2$, Markov's inequality implies that
\begin{align} \label{eqn:fixed-step1}
\mathrm{Pr}\big[ \tfrac{1}{2} \big\| \mathcal{U}(\rho) - \mathcal{V}_N \circ \cdots \mathcal{V}_1(\rho) \big\|_1
    > \epsilon \big]
    &\leq \epsilon^{-q} \mathbb{E} \big[ 2^{-q}
    \big\| \mathcal{U}(\rho) - \mathcal{V}_N \circ \cdots \mathcal{V}_1(\rho) \big\|_1^q \big].
\end{align}
The right-hand side of this equation is a convex function of the state $\rho$.
Thus, the maximum over all states is attained
at a pure state.
As a consequence, we can establish both claims in the proposition by limiting
our attention to an (unknown) pure state  $\rho = |\psi \rangle \! \langle \psi|$
that does not depend on the random unitary channels $\mathcal{V}_k (X) = V X V^\dagger$.

Next, we convert the trace distance of the output states into a Euclidean
distance on the state vectors themselves.  The power $q \geq 2$ will remain
fixed until the last step of the argument.  Lemma~\ref{lem:diamond-to-operator} implies
\begin{align}
& \big( \mathbb{E} \big[ 2^{-q} \|   \mathcal{U}(|\psi \rangle \! \langle \psi|)-\mathcal{V}_N \circ \cdots \circ \mathcal{V}_1 (|\psi \rangle \! \langle \psi|) \|_1^q  \big] \big)^{1/q}
\leq \big( \mathbb{E}\| (  V_N \cdots V_1 -U )| \psi \rangle \|_{\ell_2}^q \big)^{1/q} \nonumber \\
&\qquad\qquad\leq 2 \max\big\{ \big\| \underset{\text{deterministic bias $|\psi_{\mathrm{bias}}\rangle$}}{\underbrace{\left( \mathbb{E} \left[ V_N \cdots V_1 \right] - U \right)|\psi \rangle }} \big\|_{\ell_2}, \ \big( \mathbb{E}  \big\|  \underset{\text{random fluctuation $|\psi_{\mathrm{rand}} \rangle$}}{\underbrace{\left( V_N \cdots V_1 - \mathbb{E} \left[ V_N \cdots V_1 \right] \right)| \psi \rangle}} \big\|_{\ell_2}^q \big)^{1/q} \big\}.
\label{eq:fixed-input-conversion}
\end{align}
The last bound follows from the triangle inequality
and $(a+b)^q \leq 2^q \max\{ a^q, b^q \}$ for $a, b \geq 0$.

We have split up the difference into two components, a deterministic bias
and a random fluctuation.
To control the deterministic bias, we simply apply Proposition~\ref{prop:bias1}:
\begin{equation}
\| (\mathbb{E} \left[ V_N \cdots V_1 \right] - U) |\psi \rangle \|_{\ell_2}
= \| ((\mathbb{E}V)^N-U)| \psi \rangle \|_{\ell_2} \leq \| (\mathbb{E}V)^N-U \| \leq \frac{(t \lambda)^2}{N}.
\label{eq:fixed-input-bias}
\end{equation}
We will see that the bias is always negligible in comparison with the fluctuation.
To control the second term, we need the following lemma.

\begin{lemma} \label{lem:fixed-input-fluctuation}
Let $V_N,\ldots,V_1$ be i.i.d.~unitaries that implement the \textsc{qDRIFT} protocol \eqref{eq:qDRIFT} with parameters $t$ and $\lambda$.
Then, for any $q \geq 2$,
\begin{equation*}
\big(\mathbb{E} \| (V_N \cdots V_1 - \mathbb{E}[V_N \cdots V_1])| \psi \rangle \|_{\ell_2}^q \big)^{1/q} \leq 2 \sqrt{ \frac{(q-1)(t \lambda)^2}{N}}.
\end{equation*}
\end{lemma}

\noindent
We will establish this lemma below. 
\AC{The basic idea behind the proof is to express the random vector
using a martingale sequence, similar to the matrix case. We could have already call vector martingale tail bounds (Fact~\ref{fact:vector_martingale}) to arrive at the desired results. However, we will demonstrate the same results via markov's inequality and moment bounds (uniform smoothness, Proposition~\ref{prop:pythagoras}) which we introduced in Appendix~\ref{sec:intro_martingale}. }

Introduce the inequalities from~\eqref{eq:fixed-input-bias} and Lemma~\ref{lem:fixed-input-fluctuation}
into the bound~\eqref{eq:fixed-input-conversion}.  We obtain
\begin{align} \label{eq:fixed-input-Lq}
\big( \mathbb{E} \big[ 2^{-q} \|   \mathcal{U}(|\psi \rangle \! \langle \psi|)-\mathcal{V}_N \circ \cdots \circ \mathcal{V}_1 (|\psi \rangle \! \langle \psi|) \|_1^q \big] \big)^{1/q}  \leq 4 \sqrt{\frac{(q-1) (t\lambda)^2}{N}}.
\end{align}
We have used the assumption that $N \geq (t\lambda)^2$ to see that the second branch
of the maximum always dominates the first.

We may now complete the proof.
To obtain the expectation bound, we set $q = 2$ in \eqref{eq:fixed-input-Lq}
and apply Lyapunov's inequality.
To obtain the probability bound, we combine \eqref{eqn:fixed-step1} and \eqref{eq:fixed-input-Lq}
to arrive at
\begin{equation*}
\mathrm{Pr}\big[ \tfrac{1}{2} \big\| \mathcal{U}(\rho) - \mathcal{V}_N \circ \cdots \mathcal{V}_1(\rho) \big\|_1
    > \epsilon \big]
    \leq \left( \frac{16q (t\lambda)^2}{\epsilon^2 N} \right)^{q/2}.
\end{equation*}
Select $q = (\epsilon^2 N)/(16\mathrm{e}t^2 \lambda^2)$ to obtain the stated result.
The resulting probability bound is vacuous unless $q \geq 2$.
\end{proof}

\subsubsection{Proof of Lemma~\ref{lem:fixed-input-fluctuation}}

In this section, we establish the bound on the size of the fluctuations.

\begin{proof}[Proof of Lemma~\ref{lem:fixed-input-fluctuation}]
Fix a vector $| \psi \rangle$, and introduce a sequence of random vectors: $|\psi_k \rangle = \prod_{i=1}^k V_i | \psi \rangle$ for $1\leq k \leq N$.  As a consequence, $(V_N \cdots V_1 - \mathbb{E}[V_N \cdots V_1])|\psi \rangle = |\psi_N \rangle - \mathbb{E}[| \psi_{N} \rangle]$. 
We can recast this difference as a sum of two random vectors that are conditionally orthogonal in expectation:
\begin{align*}
\mathbb{E} \| | \psi_{N}\rangle - \mathbb{E}[ | \psi_N \rangle] \|_{\ell_2}^q 
&= \mathbb{E} \| (V_N - \mathbb{E}[V_N])|\psi_{N-1}\rangle + \mathbb{E}[V_N]( |\psi_{N-1}\rangle - \mathbb{E} \left[ | \psi_{N-1}\rangle \right] ) \|_{\ell_2}^q =: \mathbb{E} \|y + x \|_{\ell_2}^q.
\end{align*}
Indeed, $\mathbb{E}[y | x]=\mathbb{E}[V_N- (\mathbb{E} V_N)] |\psi_{N-1}\rangle = 0$.
We can apply uniform smoothness (Proposition~\ref{prop:pythagoras}) to split up the contributions:
\begin{align*}
(\mathbb{E} \| | \psi_{N}\rangle - \mathbb{E}[ | \psi_N \rangle] \|_{\ell_2}^q)^{2/q}&\leq (q-1) (\mathbb{E} \| (V_N - (\mathbb{E}V_N))|\psi_{N-1}\rangle\|_{\ell_2}^q)^{2/q} \\ &\qquad + (\mathbb{E} \| (\mathbb{E} V) (|\psi_{N-1}\rangle - \mathbb{E} \left[ | \psi_{N-1}\rangle \right]) \|_{\ell_2}^q)^{2/q} \\
&\leq (q-1) (\mathbb{E}\| V_N - \mathbb{E}[V_N]\|^q)^{2/q} + (\mathbb{E} \| | \psi_{N-1}\rangle - \mathbb{E} \left[ | \psi_{N-1}\rangle \right]\|_{\ell_2}^q)^{2/q}.
\end{align*}
We can now iterate this argument to conclude that 
\begin{align*}
(\mathbb{E} \| | \psi_{N}\rangle - \mathbb{E}[ | \psi_N \rangle] \|_{\ell_2}^q)^{2/q}\leq (q-1) \sum_{k=1}^N (\mathbb{E}\|V_k - \mathbb{E}[V_k]\|^q)^{2/q}
= (q-1) N  (\mathbb{E}\|V - \mathbb{E}[V] \|^q)^{2/q}
\end{align*}
Invoke Lemma~\ref{lem:taylor}, using the properties of the random unitaries constructed by \textsc{qDRIFT}:
\begin{equation*}
    (\mathbb{E}\|V - \mathbb{E}[V] \|^q)^{2/q} \leq \left(2 \frac{t \lambda}{N}\right)^2.
\end{equation*}
Combine the last two displays to reach the stated result.
\end{proof}

\subsection{Proof of Theorem~\ref{thm:summary_of_errors} and Corollary~\ref{cor:randomsuzuki}}\label{sec:proofgeneral}
\AC{
The proof techniques for establishing Theorem~\ref{thm:allinput} and Theorem~\ref{thm:fixedinput} are remarkably general and we condense them into Theorem~\ref{thm:summary_of_errors}. Let us reiterate the premise. Consider approximating a target unitary product $U=U_N\cdots U_1$ by a random unitary $V=V_N\cdots V_1$ such that $\{ V_k\}$ satisfy:\newline
(i) \textit{Causality:} for $1 \leq k \leq N$ the random selection of $V_k$ can only depend on previous choices for $V_1,\ldots,V_{k-1}$:
\begin{align*}
\mathrm{Pr} \left[ V_k |V_N,\ldots,V_{k+1} ,V_{k-1},\ldots, V_1 \right]
=& \mathrm{Pr} \left[ V_k| V_{k-1},\ldots,V_1 \right] \nonumber 
\end{align*}
(ii) \textit{Accurate approximation:} each realization of $V_k$ (and their conditional expectation) must be close to the ideal unitary $U_k$. More precisely, let $R,b_k >0$ be constants such that \begin{align*}
\lV V_k-\BE_{k-1} V_k\rV\le R \quad \text{and} \quad 
\left\| \mathbb{E}_{k-1} V_k - U_k \right\| \leq b_k, \text{where}\ \ \ \ \mathbb{E}_{k-1}V_k:=\mathbb{E} \left[ V_k| V_{k-1},\ldots,V_1 \right] \nonumber.
\end{align*}
almost surely for all $1 \leq k\leq N$.
Note this is more general then we need to prove Theorem~\ref{thm:summary_of_errors}; this is to take into account the cases when the conditional variance may be much smaller than $NR^2$.  
}
\AC{
\begin{proof}[Proof of Theorem~\ref{thm:summary_of_errors}]
Recall the decomposition of approximation error into a deterministic bias and a random fluctuation:
\begin{align*}
\| V_N \cdots V_1 - U_N\cdots U_1 \| 
&\leq  \underset{\text{deterministic bias}}{\underbrace{ \| (\mathbb{E} V )^N - U_N\cdots U_1 \| }}
+
\underset{\text{random fluctuation}}{\underbrace{\| V_N \cdots V_1 - \mathbb{E} \left[ V_N \cdots V_1 \right] \| }}.
\end{align*}
The deterministic bias can be once again be controlled by a telescoping sum,
\begin{align}
  \| (\mathbb{E} V )^N - U \|   &\le \sum_{k=1}^N b_k.
\end{align}
Note that this also controls the performance of channel $\lV \BE \CV - \CU \rV_\diamond$ by Lemma~\ref{lem:diamond-to-operator}.
 For the \textit{random fluctuations}, tweaking the proofs of Theorem~\ref{thm:allinput} and Theorem~\ref{thm:fixedinput} implies the following general result.

\begin{proposition}[Adapted random walk on the unitary group; with and without fixed input state]\label{prop:concentration_unitary_gp}
Let $\{V_1, V_2, \cdots, V_N \} \subset U(d)$ an adapted(causal) random unitary matrices that obey
\begin{align}
  \sum^N_{k=1}  \BE_{k-1} \lVert(V_k - \BE V_k)(V^\dagger_k - \BE V^\dagger_k)\rVert \le \sigma^2
  \quad\text{and}\quad
  \lVert V_k - \BE V_k \rVert \le R. (\text{almost surely.})
\end{align}
for some constants $v, R \geq 0$.
Then,  the product of $N$ random unitaries satisfies a concentration inequality:
\begin{equation}\label{ldprw}
\mathrm{Pr}(\lVert V_N\cdots V_1 - \BE[ V_N \cdots V_1] \rVert \ge \epsilon  ) \le 2d \exp\left( \frac{-\epsilon^2/2}{ v+ R\epsilon/3}\right) \quad \text{for $\epsilon >0$.}
\end{equation}
For a fixed, but arbitrary, input state $\rho$, concentration is independent of the ambient dimension $d$:
\begin{equation*}
\max_{\rho\ \mathrm{state}}
    \mathrm{Pr}\big( \tfrac{1}{2} \big\| \BE[\mathcal{V}_N \circ \cdots \mathcal{V}_1(\rho)] - \mathcal{V}_N \circ \cdots \mathcal{V}_1(\rho) \big\|_1
    > \epsilon \big)
    \leq 2\exp\left( \frac{-\epsilon^2/2 }{v+R\epsilon/3} \right) \quad \text{for $\epsilon >0$}.
\end{equation*}
\end{proposition}

For a fixed input state, we would call the vector Freedman inequality (Fact~\ref{fact:vector_martingale}) instead of uniform smoothness (Proposition~\ref{prop:pythagoras}) in Theorem~\ref{thm:fixedinput}. 
There are several recent independent papers that also use matrix martingale tools to study products of random matrices that are close to the identity.  The work~\cite{HNTR20:Matrix-Product} addresses the problem using uniform smoothness tools.  The paper~\cite{kathuria2020concentration} uses the matrix Freedman inequality; their proof is quite similar to ours.  In contrast, we are interested in unitary products, which allows for additional simplifications.  For more background on matrix martingales, see~\cite{oliveira2009spectrum,tropp2011freedman,HNTR20:Matrix-Product,Christofides2008martingales}.

It is instructive to illustrate these improvements by example. In \textsc{qDRIFT}, all steps have a uniform bound $R$, but in the fully general statement the variance $v$ can differ from the crude uniform bound $NR^2$. In such a regime, the sub-exponential tail of size $\e^{-3\epsilon/2R}$ can start playing a role.

Lastly, for illustration in Theorem~\ref{thm:summary_of_errors} we give loose estimates on the variance to avoid complication with the heavy tail effects. Plugging in the parameters $v = ra^2, R = 2a$ as $\lV V-\BE V \rV\le \BE\lV V -V'\rV \le 2 \lV U-V\rV = 2a$ translates to a typical fluctuation $\epsilon^2 \sim n ra^2$(and $\epsilon^2 \sim ra^2$ for fixed input). We conclude the proof by combining the bound for the deterministic bias and random fluctuation. 

\end{proof}
}
\begin{proof}[Proof of Corollary~\ref{cor:randomsuzuki}]
Consider randomly permuting the 2k-th order Suzuki formulas as
$V_k = S^{\sigma}_{2k}(t/r)$ with uniform probability $1/L!$. Then by direct calculation~\cite{childs2019faster, suzuki1985decomposition}: 
\begin{align}
    \lV V-U\rV &\le O(\frac{(t\Lambda L)^{2k+1}}{r^{2k+1}})=a ,\\
\lV \BE V -U\rV &\le  O(\frac{t\Lambda}{r})^{2k+1}L^{2k})=b
\end{align}
by Theorem~\ref{thm:summary_of_errors},
\begin{align*}
    \epsilon_{\mathrm{det}} & \le ra =O( \frac{(t\Lambda L)^{2k+1}}{r^{2k}} )\\
    \epsilon_{\mathrm{typ}} &= O(\sqrt{nr}a)+2rb = O(\sqrt{n} \frac{(t\Lambda L)^{2k+1}}{r^{2k+1/2}} +(t\Lambda)^{2k+1}(\frac{L}{r})^{2k})\\
     \epsilon_{fix} &= O(\sqrt{r}a)+2rb = O( \frac{(t\Lambda L)^{2k+1}}{r^{2k+1/2}}+(t\Lambda)^{2k+1}(\frac{L}{r})^{2k}) \\
    \epsilon_{\mathrm{avg}} &\le 2rb  =O(  (t\Lambda)^{2k+1}(\frac{L}{r})^{2k})
    \end{align*}
Which translates to the sufficient gate counts $N =rL$
    \begin{align*}
    N_{\mathrm{det}} & \gtrsim t\Lambda L^2 (\frac{t\Lambda L}{\epsilon})^{1/2k}\\
    N_{\mathrm{typ}} & \gtrsim t\Lambda L^2(\frac{n\Lambda Lt}{\epsilon^2})^{1/(4k+1)}+t \Lambda L^2 (\frac{t\Lambda}{\epsilon})^{1/2k}\\
    N_{\mathrm{fix}} & \gtrsim t\Lambda L^2(\frac{\Lambda Lt}{\epsilon^2})^{1/(4k+1)}+t \Lambda L^2 (\frac{t\Lambda}{\epsilon})^{1/2k}\\
    N_{\mathrm{avg}}&\gtrsim  t \Lambda L^2 (\frac{t\Lambda}{\epsilon})^{1/2k} 
\end{align*}
\end{proof}

\section{Asymptotic tightness} \label{sec:tightness}

It is natural to wonder whether the bound~\eqref{eq:expected-error} is tight for some Hamiltonian $H=\sum_j h_j$, i.e., whether $N = \Omega( n \lambda^2 t^2/\epsilon^2)$ is also necessary to achieve concentration.
More precisely, we want to understand whether the dependences on system size $n = \log_2(d)$, evolution time $t$ and interaction strength $\lambda=\sum_j \norm{h_j}$ are also necessary to control the typical deviation of the unitary random walk we considered.

In the context of matrix concentration inequalities, this question has been thoroughly addressed~\cite[Section~4.1.2]{tropp2015introduction}. The answer is affirmative for sums of bounded matrices: concentration inequalities are tight and saturated for collections of \textit{commuting} matrices. Although in this work we consider products of random matrices,  we are still using a telescoping sum in the small step regime and expect an analogy.

This observation motivates us to look at artificial Hamiltonians whose associated unitary evolution saturates the upper bounds put forth in this work.
The cases we can handle lie at the two extremes: either the sum of single-site Pauli $Z$s or the sum of all $2^n$ many-body Pauli $Z$s.  We will see the presence of 
the system size factor $n = \log_2(d)$ at both extremes, so one may believe the same to hold for the intermediate q-local cases. 
However, this factor arises for very different reasons.
It arises in the single-site case, because the operator norm completely factorizes into $n$ constituents (one for each term).
For Hamiltonians that encompass all $2^n$ many-body $Z$s, it comes from the fact that diagonal entries are nearly independent, so the union bound we used in Section~\ref{unionboundexp} is tight.  Independence of entries requires the presence of all many-body terms, and does not extend to the few-body case. 

The multivariate central limit Theorem will be crucial for analyzing both cases, as it greatly simplifies the analysis in large $N$ limit.

\begin{fact}[CLT for the multinomial distribution]\label{fact:mCLT}

The multinomial distribution $\vct{m}=(m_1,\ldots,m_K)\sim \mathrm{Mult}(N,(1/K,\ldots,1/K))$ (roll a fair $K$-sided dice $N$ times) obeys a 
central limit theorem (CLT):
\begin{equation*}
\tfrac{1}{\sqrt{N}}(\vct{m}-\mathbb{E} \vct{m}) \sim \mathcal{N}(0,\Sigma) \quad \text{\AC{ in distribution as $N \to \infty$}}.
\end{equation*}
The covariance matrix is
$\Sigma = \tfrac{1}{K}\left(\mathbb{I}-\tfrac{1}{K}J\right)$, where $J$ denotes the $K \times K$ matrix of ones. 
\end{fact}

\subsection{Sum of single site Pauli-Z operators}\label{singlesite}

This example demonstrates the saturation of our martingale bounds for single site Hamiltonians that factorize completely.
To this end, we revisit a variant of the $n$-qubit example Hamiltonian discussed in Section~\ref{nvs1}:
\begin{equation}
H= \sum_{k=1}^n Z_k \quad \text{where} \quad Z_k = \underset{\text{$(k-1)$-times}}{\underbrace{\mathbb{I} \otimes \cdots \otimes \mathbb{I}}} \otimes Z \otimes 
\underset{\text{$(n-k)$-times}}{\underbrace{\mathbb{I} \otimes \cdots \otimes \mathbb{I}}} \quad \text{for $1 \leq k \leq n$}. \label{eq:single-site}
\end{equation}

\begin{proposition} \label{prop:single-site}
Suppose that we wish to obtain an $N$-term approximation of the time evolution $U = \exp(-\mathrm{i}tH)$ associated with the $n$-qubit Hamiltonian \eqref{eq:single-site} for evolution time $t$. In the large $N$ limit (CLT), the \textsc{qDRIFT} approximation \eqref{eq:qDRIFT} incurs an operator norm error that matches the (upper) bound from Corollary~\ref{cor:exp_allinput} (and indirectly Theorem~\ref{thm:allinput}) up to a constant factor:
\begin{equation*}
\mathbb{E} \norm{U-V_N \cdots V_1} \geq \sqrt{\frac{2}{\pi}} \sqrt{(n-1) \frac{(t\lambda)^2}{N}}- \frac{1}{2}(n-1) \frac{(t \lambda)^2}{N}.
\end{equation*}
\end{proposition}

We have chosen to state this result directly in terms of operator norm deviation. A conversion into diamond distance is also possible: $\tfrac{1}{2}\|\mathcal{U}-\mathcal{V}\|_\diamond \geq \tfrac{1}{2}\|U-V\|$ for any pair of unitary channels. 
This conversion rule readily follows from the geometric characterization of $\tfrac{1}{2}\|\mathcal{U}-\mathcal{V}\|_\diamond$ provided in~\cite{AKN98:Quantum-Circuits}.

\begin{proof}[Proof of Proposition~\ref{prop:single-site}]
Each of the $n$ terms in the Hamiltonian \eqref{eq:single-site} has unit operator norm ($\norm{Z_k}=1$) and the total strength is $\lambda=\sum_{k=1}^n \norm{Z_k}=n$. For fixed $N$ and $t$, each short-time approximation \eqref{eq:qDRIFT} has the form $V_k= \exp\left( -\mathrm{i} \tfrac{tn}{N}Z_{k(i)}\right)$, where each $k(i)$ is an index chosen uniformly from the set $\left\{1,\ldots,n\right\}$ (multinomial distribution). Since all $Z_k$s commute, we can rewrite the entire product formula as
\begin{equation*}
V_N \cdots V_1 = \exp \left( - \mathrm{i} \frac{t\lambda}{N}\sum_{i=1}^N Z_{k(i)} \right)
= \exp \left( - \mathrm{i} \frac{t\lambda}{N}\sum_{k=1}^n m_k Z_k \right).
\end{equation*}
Here, we have introduced the count statistics $m_k$ for each site label $k$, that is the number of times location $k$ has been selected throughout $N$ independent selection rounds, to rearrange the sum.
This count statistics obeys $\bar{m}_k = \mathbb{E} m_k = N/n=N/\lambda$ for each $1 \leq k \leq n$. We can use this observation to re-express the target unitary $U$ in a compatible fashion:
\begin{equation*}
U = \exp \left( - \mathrm{i}t \sum_{k=1}^n Z_k \right) = \exp \left( - \mathrm{i} \frac{t\lambda}{N}\sum_{k=1}^n \bar{m}_k Z_k \right).
\end{equation*}
Unitary invariance then implies that the operator norm difference between both unitaries becomes
\begin{equation}
\norm{V_N \cdots V_1 - U} = \norm{ \exp \left( - \mathrm{i} \tfrac{t\lambda}{\sqrt{N}}\sum_{k=1}^N \frac{m_k - \bar{m}_k}{\sqrt{N}} Z_k\right) - \mathbb{I}}. \label{eq:single-site-aux1}
\end{equation}
This is a promising starting point. The multinomial CLT (Fact~\ref{fact:mCLT}) ensures that the $n$ centered and normalized random variables $s_k = (m_k-\bar{m}_k)/\sqrt{N}$ approach the coefficients of a Gaussian vector $s \in \mathbb{R}^n$ with covariance matrix $\Sigma = \tfrac{1}{n}\left( \mathbb{I} - \tfrac{1}{n}J \right)$.
This, in particular implies $\mathbb{E} s_k = 0$ and $\mathbb{E}s_k^2 = \tfrac{1}{n}(1-\tfrac{1}{n})=\sigma^2$ for all $1 \leq k \leq n$. We can capitalize on this observation by simplifying \eqref{eq:single-site-aux1} via a second-order Taylor expansion. Set $X = -\tfrac{t\lambda}{\sqrt{N}}\sum_k s_k Z_k$ for brevity and apply Fact~\ref{fact:taylor} to obtain
\begin{equation*}
\norm{V_N \cdots V_1 - U} = \norm{ \exp(\mathrm{i}X)-\mathbb{I}} \geq \norm{\mathrm{i}X} - \norm{ \exp(\mathrm{i}X)-\mathrm{i}X - \mathbb{I}} \geq \norm{X} - \tfrac{1}{2}\norm{X}^2.
\end{equation*}
This relation is preserved under expectations and we obtain
\begin{equation*}
\mathbb{E} \norm{V_N \cdots V_1 - U} \geq \frac{t\lambda}{\sqrt{N}}\mathbb{E} \norm{\sum_k s_k Z_k} - \frac{1}{2}\left( \frac{t\lambda}{\sqrt{N}}\right)^2 \mathbb{E} \norm{\sum_k s_k Z_k}^2.
\end{equation*}
Let us focus on the leading order term first. The particular structure of the Hamiltonian \eqref{eq:single-site} -- each $Z_k$ is the tensor product of a single Pauli-Z matrix at location $k$ with $(n-1)$ identity matrices -- ensures that the operator norm factorizes nicely. Use $\norm{X \otimes \mathbb{I} + \mathbb{I} \otimes Y} = \norm{X}+\norm{Y}$ iteratively to conclude
\begin{equation*}
\mathbb{E} \norm{\sum_k s_k Z_k} = \mathbb{E} \sum_{k=1}^n \norm{s_k Z_k} = \sum_{k=1}^n |s_k| \overset{N \to \infty}{=} n\sqrt{\frac{2}{\pi}\frac{1}{n}\left(1-\frac{1}{n}\right)}=\sqrt{\frac{2}{\pi}(n-1)},
\end{equation*}
because the CLT asserts that each $|s_k|$ approaches a half-normal random variable with $\sigma^2 = \tfrac{1}{n}(1-\tfrac{1}{n})$.

To bound the quadratic term, we combine the factorization trick from above with a well-known relation among $\ell_p$-norms in $\mathbb{R}^n$: 
\begin{align*}
\mathbb{E} \norm{\sum_{k=1}^n s_k Z_k}^2 = \mathbb{E} \left( \sum_{k=1}^n |s_k| \right)^2 
= \mathbb{E} \norm{s}_{\ell_1}^2 \leq n \mathbb{E} \norm{s}_{\ell_2}^2 = n \sum_{k=1}^n \mathbb{E} s_k^2 = n^2 \sigma^2 = (n-1).
\end{align*}
No CLT is required for this argument.
Inserting both bounds into Eq.~\eqref{eq:single-site-aux1} completes the argument.
\end{proof}

\subsection{Sum of many-body Pauli-Z operators}
\label{sub:many-sites}

Let us revisit the example Hamiltonian from Sec.~\ref{unionboundexp}, albeit without additional sign factors. Recall the multi-indices $\vct{p} =(p_1,\ldots,p_n) \in \left\{0,1\right\}^n$  and set
\begin{equation}
H = \sum_{\mathbf{p} \in \left\{0,1\right\}^n} Z_{\mathbf{p}}
= \sum_{\mathbf{p} \in \left\{0,1\right\}^n} Z^{p_1}\otimes \cdots \otimes Z^{p_n},
\label{eq:saturation-hamiltonian}
\end{equation}
where we use the conventions $Z^1=Z$ and $Z^0 = \mathbb{I}$.
This Hamiltonian is not local.  All constituents commute and have the same operator norm: $\| Z_{\mathbf{p}}\|=1$ for all $\mathbf{p}\in \left\{0,1\right\}^n$.
This in turn implies that the total strength $\lambda=\sum_{\mathbf{p}} \| Z_{\mathbf{p}}\|=2^n$ equals the Hilbert space dimension. It is also worthwhile to point out that each term is diagonal in the computational basis $|\mathbf{b} \rangle = |b_1,\ldots,b_n \rangle$ with $\mathbf{b}=(b_1,\ldots,b_n)\in \left\{0,1\right\}^n$. Overlaps of the Hamiltonian terms with computational basis states are given by
\begin{equation}
\langle \mathbf{b} | Z_{\mathbf{p}} | \mathbf{b} \rangle = (-1)^{\langle \mathbf{b},\mathbf{p}\rangle} = (-1)^{\sum_i b_i p_i} \in \left\{\pm 1 \right\}.
\label{eq:combinatorial}
\end{equation}
The following claim highlights that our findings are tight for asymptotically large step sizes $N$. This complements the example upper bound derived in Sec.~\ref{unionboundexp}, as well as Theorem~\ref{thm:allinput}.

\begin{proposition} \label{prop:many-sites}
Suppose that we wish to obtain an $N$-term approximation of the time evolution $U = \exp( -\mathrm{i}t H)$ associated with the $n$-qubit Hamiltonian \eqref{eq:saturation-hamiltonian} for evolution time $t$.
In the large $N$ limit (CLT) the \textsc{qDRIFT} approximation
\eqref{eq:qDRIFT} incurs an operator norm error that matches the (upper) bound from Corollary~\ref{cor:exp_allinput} (and indirectly Theorem~\ref{thm:allinput}) up to a constant factor:
\begin{equation*}
\mathbb{E} \|U - V_N \cdots V_1 \|
\geq \frac{1}{2} \sqrt{n \frac{(t\lambda)^2}{N}} -  2 \left(n+\tfrac{1}{2}\right) \frac{(t \lambda)^2}{N}
\end{equation*}
\end{proposition}

The conversion rule $\norm{\mathcal{U}-\mathcal{V}}_\diamond \geq \norm{U-V}$ (for unitary channels) \cite{AKN98:Quantum-Circuits} once more allows for addressing the expected diamond distance as well.

\begin{proof}[Proof of Proposition~\ref{prop:many-sites}]
Each of the $2^n$ terms in the Hamiltonian~\eqref{eq:saturation-hamiltonian}  has unit operator norm ($\|Z_{\vct{p}}\|=1$ for all $\vct{p} \in \left\{0,1\right\}^n$) and the strength is $\lambda = \sum_{\vct{p}}\|Z_{\vct{p}}\|=2^n$. 
For fixed $N$ and $t$,
 each short-time approximation \eqref{eq:qDRIFT} has the form $V_k = \exp (-\mathrm{i} \tfrac{t \lambda}{N} Z_{\vct{p}(i)})$, where $\vct{p}(i)$ is a string chosen uniformly at random from all $2^n$ possibilities (multinomial distribution). Since all $Z_{\vct{p}}$s commute,
 we can rephrase and simplify the expected operator norm difference in a fashion analogous to the proof of Proposition~\ref{prop:single-site}:
\begin{equation}
\mathbb{E} \norm{V_N \cdots V_1 - U}
\geq
\frac{t\lambda}{\sqrt{N}} \mathbb{E}\norm{ \sum_{\vct{p}}s_{\vct{p}}Z_{\vct{p}}}-\frac{1}{2}\left( \frac{t \lambda}{\sqrt{N}} \right)^2 \mathbb{E}\norm{\sum_{\vct{p}}s_{\vct{p}}Z_{\vct{p}}}^2. \label{eq:many-sites-aux2}
\end{equation}
Here, $s_{\vct{p}}=(m_{\vct{p}}-\bar{m}_{\vct{p}})/\sqrt{N}$ is the centered and normalized variant of the count statistics $m_{\vct{p}}$ associated with bit string $\vct{p} \in \left\{0,1\right\}^n$, that is the number of times the Hamiltonian term $Z_{\vct{p}}$ has been selected throughout $N$ independent selection rounds. 
The multinomial CLT (Fact~\ref{fact:mCLT}) asserts that the $2^n$ centered and normalized random variables $s_{\vct{p}}$ approach distinct coefficients of a $2^n$-dimensional Gaussian vector with covariance matrix 
 $\Sigma = \tfrac{1}{2^n} \left(\mathbb{I}-\tfrac{1}{2^n} J \right)=\tfrac{1}{2^n}(\mathbb{I}-|\vct{1}\rangle \! \langle \vct{1}|)$, where $|\vct{1} \rangle = \tfrac{1}{2^n}\sum_{\vct{b} \in \left\{0,1\right\}^n}|\vct{b} \rangle$ (the normalized all-ones vector in $\mathbb{R}^{2^n}$).
 In contrast to before, the individual contributions to this operator norm don't factor nicely anymore. Establishing tight bounds requires additional analysis.

Let us focus on the (leading) first-order term for now. All matrix summands in the expression commute and are diagonal in the computational basis $|\vct{b} \rangle$ with $\vct{b}=(b_1,\ldots,b_n) \in \left\{0,1\right\}^n$. This ensures that the operator norm is attained at a computational basis state:
\begin{equation*}
\norm{\sum_{\vct{p}} Z_{\vct{p}}s_{\vct{p}}} = \max_{\vct{b} \in \left\{0,1\right\}^n} \left| \sum_{\vct{p}} s_{\vct{p}} \langle \vct{b}|Z_{\vct{p}}| \vct{b} \rangle \right| = \max_{\vct{b} \in \left\{0,1\right\}^n} \left| \sum_{\vct{p}} (-1)^{\langle \vct{b},\vct{p}\rangle} s_{\vct{p}} \right|,
\end{equation*}
where the last equation is due to Rel.~\eqref{eq:combinatorial}. 
This expression is proportional to the largest entry (in modulus) of the Walsh-Hadamard transform of the $2^n$-dimensional vector $s$ with entries $s_{\vct{p}}$ for $\vct{p} \in \left\{0,1\right\}^n$. More precisely,
\begin{equation*}
\max_{\vct{b}\in \left\{0,1\right\}^n} \left| \sum_{\vct{p}} (-1)^{\langle \vct{b},\vct{p} \rangle} s_\vct{p} \right| =  2^{n/2}\norm{ \mathrm{Had}^{\otimes n} s }_{\ell_\infty}=: 2^{n/2} \norm{\hat{s}}_{\ell_\infty} \quad \text{where} \quad \mathrm{Had} = \tfrac{1}{\sqrt{2}} \left(
\begin{array}{cc}
1 & 1 \\
1 &-1
\end{array}
\right).
\end{equation*}
We emphasize that the Walsh-Hadamard transform is an orthogonal transformation, which also applies to the limiting covariance matrix of $\hat{s}=\mathrm{Had}^{\otimes n}s$ (CLT):
\begin{equation*}
\hat{\Sigma} = \tfrac{1}{2^n}\mathrm{Had}^{\otimes n}  \left(\mathbb{I}-|\vct{1} \rangle \! \langle \vct{1}| \right) \mathrm{Had}^{\otimes n}
=\tfrac{1}{2^n}\left( \mathbb{I} - |0,\ldots,0 \rangle \! \langle 0,\ldots,0| \right).
\end{equation*}
Hence, the CLT asserts that the transformed vector $\hat{s}$ approaches a standard Gaussian vector with $2^n-1$ degrees of freedom: $\hat{s}=(0,g_2,\ldots,g_{2^n})^T$ with $g_i \overset{\text{i.i.d.}}{\sim}\mathcal{N}(0,2^{-n})$ (one degree of freedom is erased by the normalization constraint $\sum_{\vct{p}}m_{\vct{p}}=N$ of the count statistics).
The bound on the expected leading order contribution now follows from invoking the well-known fact that the expected maximum of $K$ standard Gaussian random variables with equal variance $\sigma^2$ is lower-bounded by $0.265\sqrt{\log(K)\sigma^2}$, see e.g.\ \cite[Proposition~8.1]{FR13:Compressed-Sensing}:
\begin{align*}
\tfrac{t\lambda}{\sqrt{N}}\mathbb{E} \norm{ \sum_{\vct{p}}s_{\vct{p}}Z_{\vct{p}}}
=& \frac{t \lambda }{\sqrt{N}}2^{n/2} \mathbb{E} \|\hat{s} \|_{\ell_\infty}\overset{N \to \infty}{=} \frac{t \lambda }{\sqrt{N}} 2^{n/2}\mathbb{E} \max_{2 \leq i \leq 2^n}|g_i| \\ \geq & 0.625 \frac{t\lambda}{\sqrt{N}}2^{n/2} \sqrt{\log(2^n-1)2^{-n}}  \geq \frac{1}{2}\sqrt{n \frac{(t\lambda)^2}{N}}.
\end{align*}
Here, we have used the numerical bound $0.625\sqrt{\log(2^n-1)/n} \geq 0.5$ which is valid for any $n \geq 3$ (for $n=2$ the ratio is slightly smaller).
This completes the argument for the leading term in Eq.~\eqref{eq:many-sites-aux2}.

Moving on to the quadratic term in Eq.~\eqref{eq:many-sites-aux2}, we employ a similar strategy. Observe
\begin{align*}
 \mathbb{E} \norm{ \sum_{\vct{p}} Z_{\vct{p}}}^2
 =& \mathbb{E} \norm{\sum_{\vct{p},\vct{p}'} s_{\vct{p}} s_{\vct{p}'} Z_{\vct{p}} Z_{\vct{p}'}}
= \mathbb{E} \max_{\vct{b} \in \left\{0,1\right\}^n} \left| \sum_{\vct{p},\vct{p'}} s_{\vct{p}} s_{\vct{p}'} \langle \vct{b}|Z_{\vct{p}} Z_{\vct{p}'} |\vct{b} \rangle \right| \\
=& \mathbb{E} \max_{\vct{b} \in \left\{0,1\right\}^n} \left| \sum_{\vct{p}} (-1)^{\langle \vct{b},\vct{p}\rangle} s_{\vct{p}} \sum_{\vct{p}'} (-1)^{\langle \vct{b},\vct{p}'\rangle} s_{\vct{p}'} \right|,
\end{align*}
where the last equation follows from combining Eq.~\eqref{eq:combinatorial} with the appealing group structure of the $Z_{\vct{p}}$'s: $Z_{\vct{p}}Z_{\vct{p}'}=Z_{\vct{p}\oplus \vct{p}'}$, where $\oplus$ denotes entry-wise addition modulo 2 (the set of all $Z_{\vct{p}}$'s form a maximal stabilizer group). We can now recognize two independent Walsh-Hadamard transforms of the $2^n$-dimensional vector $s$ in this expression:
\begin{equation*}
\mathbb{E} \max_{\vct{b} \in \left\{0,1\right\}^n} \left| \sum_{\vct{p}} (-1)^{\langle \vct{b},\vct{p}\rangle} s_{\vct{p}} \sum_{\vct{p}'} (-1)^{\langle \vct{b},\vct{p}'\rangle} s_{\vct{p}'} \right|
= 2^n \mathbb{E} \max_{\vct{b} \in \left\{0,1\right\}^n} \left| \hat{s}_{\vct{b}}\right|^2
\end{equation*}
We already know from the CLT that the $2^n$-dimensional Walsh-Hadamard transform of $s$ approaches a standard Gaussian vector: $\hat{s}=(0,g_2,\ldots,g_{2^n})^T$ with $g_i \overset{\text{i.i.d.}}{\sim}\mathcal{N}(0,2^{-n})$.
In the large $N$ limit (CLT), the r.h.s.\ of the above display becomes an expected maximum of $K=2^n-1$ squares of i.i.d. Gaussian variables with mean zero and variance $\sigma^2=2^{-n}$. 
Such expected maxima can be bounded using standard arguments, see e.g.\ \cite[Lemma~5.1]{VH14:Probability-Lecture}: $\mathbb{E} \max_{1 \leq i \leq K} |g_i|^2 \leq 4 \sigma^2 \log (\sqrt{2}K)$ (the constants are chosen based on simplicity, not tightness).
This allows us to conclude
\begin{align*}
\mathbb{E} \norm{ \sum_{\vct{p}} Z_{\vct{p}}}^2
=2^n \mathbb{E} \max_{\vct{b}\in \left\{0,1\right\}^n} \left| \hat{s}_{\vct{b}} \right|^2 \overset{N \to \infty}{=} 2^n \mathbb{E} \max_{2 \leq i \leq N} \left|g_i \right|^2
\leq 2^n 4 \sigma^2 \log(\sqrt{2}(2^n-1)) \leq 4(n+1/2).
\end{align*}
Inserting linear and quadratic bound into Eq.~\eqref{eq:many-sites-aux2} completes the argument.
\end{proof}


\end{document}